%% file: Critical-case.tex
\documentclass[10pt]{article}
\usepackage{amsmath,amssymb,cancel, graphics}
\usepackage[colorlinks=true, pdfstartview=FitV, linkcolor=black, citecolor=black,urlcolor=black]{hyperref}

\usepackage{amsfonts}
\usepackage{amsthm}
\usepackage{verbatim}
\usepackage{enumerate}
\usepackage{color}
\usepackage{framed}

\def\insrt#1{\blue {#1}} 
\definecolor{shadecolor}{rgb}{0.9, 0.9, 0.81}

\numberwithin{equation}{section}

\newtheorem{thm}{Theorem}
\newtheorem{prop}{Proposition}[section]
\newtheorem{assumption}{Assumption}[section]

\newtheorem{defn}[prop]{Definition}
\newtheorem{lemma}[prop]{Lemma}

\newtheorem{remark}{Remark}

\newcommand{\bpm}{\begin{pmatrix}}
\def\epm{\end{pmatrix}}

\def\le{\left}
\def\ri{\right}
\def\be{\begin{equation}}
\def\ee{\end{equation}}
\def\bea{\begin{eqnarray}}
\def\eea{\end{eqnarray}}
\def\bp{\begin{prop}}
\def\ep{\end{prop}}
\def \pa{\partial}
\def\wh{\widehat}
\def\1{{\bf 1}}

\def\bs{\boldsymbol}
\def\d{d}
\def\eq{\begin{equation}}
\def\endeq{\end{equation}}
\def\eqarr{\begin{eqnarray}}
\def\endeqarr{\end{eqnarray}}
\def\eqnn{\begin{equation*}}
\def\endeqnn{\end{equation*}}
\def\ds{\displaystyle}

\def\e{\varepsilon}

\newcommand{\C}{\mathbb{C}}

\newcommand{\R}{\mathbb{R}}

\definecolor{light-blue}{rgb}{0.8,0.85,1}
\definecolor{blue}{rgb}{0,0,1}
\definecolor{red}{rgb}{1,0,0}
\def\red#1{\textcolor[rgb]{0.9, 0, 0}{#1}}

\def\blue#1{\textcolor[rgb]{0,0,1}{#1}}

\begin{document}

\baselineskip 16pt plus 1pt minus 1pt

\title{{\bf Spectra of Random Hermitian Matrices with a Small-Rank External Source:  The critical and near-critical regimes}}


\author{{\bf M. Bertola}\footnote{Department of Mathematics and Statistics, Concordia University, Montr\'eal, Qu\'ebec H3G1M8 and Centre de recherches math\'ematiques, Universit\'e de Montr\'eal, Qu\'ebec H3T1J4;  bertola@mathstat.concordia.ca},
{\bf R. Buckingham}\footnote{Department of Mathematical Sciences, University of
Cincinnati, Ohio 45221; buckinrt@uc.edu},
{\bf S. Y. Lee}\footnote{Department of Mathematics, California Institute of Technology, Pasadena, California 91125; duxlee@caltech.edu},
{\bf V. Pierce}\footnote{Department of Mathematics, University of Texas -- Pan American, Edinburg, Texas, 78539; piercevu@utpa.edu}}

\date{\today}
\maketitle

\begin{abstract}
Random Hermitian matrices are used to model complex 
systems without time-reversal invariance.  
Adding an external source to the model can have the effect of 
shifting some of the matrix eigenvalues, which corresponds to 
shifting some of the energy levels of the physical system. 
We consider the case when the $n\times n$ external 
source matrix has two distinct real eigenvalues:  $a$ with multiplicity 
$r$ and zero with multiplicity $n-r$.  For a Gaussian potential, it 
was shown by P\'ech\'e \cite{Peche:2006} that when $r$ is fixed or grows 
sufficiently slowly with $n$ (a small-rank source),  
$r$ eigenvalues are expected to exit the main bulk for $|a|$ large 
enough.  Furthermore, at the critical value of $a$ when the outliers are 
at the edge of a band, the eigenvalues at the edge are described 
by the $r$-Airy kernel.  We establish the universality of the $r$-Airy 
kernel for a general class of analytic potentials for 
$r=\mathcal{O}(n^\gamma)$ for $0\leq\gamma<1/12$.
\end{abstract}

\tableofcontents

\section{Introduction}

Fix a Hermitian matrix ${\bf A}$.  Equip the space of $n\times n$ Hermitian 
matrices ${\bf M}$ with the probability measure
\eq
\label{prob-measure}
\mu_n(d{\bf M}) = \frac{1}{Z_n}e^{-n\text{Tr}(V({\bf M})-{\bf A}{\bf M})}d{\bf M}; \quad Z_n:=\int e^{-n\text{Tr}(V({\bf M})-{\bf A}{\bf M})}d{\bf M},
\endeq
where $d{\bf M}$ is the entry-wise Lebesgue measure
and the integration is over all Hermitian matrices.  The eigenvalues 
of ${\bf M}$ represent the energy levels of a system without
time-reversal invariance \cite{Wigner:1951}.  When the \emph{external field} 
${\bf A}$ is nonzero and $V({\bf M})={\bf M}^2/2$, this measure arises in the 
study of Hamiltonians that can be written as the sum of a random matrix and a 
deterministic source matrix \cite{Brezin:1996}.

When ${\bf A}={\bs 0}$ (no external source) and 
$V({\bf M})={\bf M}^2/2$, \eqref{prob-measure}
describes the Gaussian Unitary Ensemble, or GUE.  For reasonable choices 
of $V(z)$ the spectrum tends to accumulate on fixed bands on the real axis.  
Introducing an external field ${\bf A}$ can have the 
effect of perturbing the expected position of the spectrum.  
For example, Aptekarev, Bleher, and
Kuijlaars \cite{Bleher:2004b, Aptekarev:2005, Bleher:2007} studied the 
Gaussian case when the matrix ${\bf A}$ has two eigenvalues $\pm a$, each of
multiplicity $n/2$.  When $a$ is sufficiently small (the 
\emph{subcritical} case), 
the eigenvalues of 
${\bf M}$ accumulate with probability one on a single interval, just as 
when ${\bf A}={\bs 0}$.  As $a$ increases the interval splits into two 
(the \emph{supercritical} case).  There is a transitional value of $a$ 
between these two cases (the \emph{critical} case) where the local eigenvalue 
density near where the bands are about to split is described by 
the \emph{Pearcey process}.  See \cite{McLaughlin:2007} and 
\cite{Bleher:2010} for further studies of large-rank external field 
models, and \cite{Adler:2009b} for recent results on the universality of 
the Pearcey process.

We are interested instead in \emph{small-rank} sources of the form
\eq
\label{A}
{\bf A} = \text{diag}(\underbrace{a,\dots,a}_r,\underbrace{0,\dots,0}_{n-r})
\endeq
assuming that $r(n)=\mathcal{O}(n^\gamma)$, $0\leq\gamma<1/12$.  
The ratio of $r$ to $n$, which is 
asymptotically small, will be denoted as 
\eq\label{kappa}
\kappa:=\frac{r}{n} = \mathcal{O}(n^{\gamma-1}).
\endeq
The limiting distribution of the largest
eigenvalue for such a small-rank external source in the Gaussian case, 
$V({\bf M})={\bf M}^2/2$, was studied by 
P\'ech\'e \cite{Peche:2006}.  Again three distinct behaviors were observed.  
For $a$ sufficiently close to zero, i.e. the subcritical case, the largest 
eigenvalue is expected to lie at the right band endpoint and
behave as the largest eigenvalue of an $n\times n$ GUE matrix.  For $a$ 
large enough, i.e. the supercritical case, $r$ eigenvalues are expected to 
exit the bulk and be distributed as the eigenvalues of an
$r\times r$ GUE matrix.  In the transitional critical case, when the outliers
lie very near the band endpoint, the distribution for the largest eigenvalue
is an extension of the standard GUE Tracy-Widom function
\cite{Tracy:1994} when $r=0$ (see also \cite{Baik:2005,Baik:2006,Adler:2009a}).
These functions, denoted by $F_r(x)$, were first discovered by
Baik, Ben-Arous, and P\'ech\'e \cite{Baik:2005} in the context of sample
covariance, or Wishart, matrices and were shown to be probability
distributions by Baik \cite{Baik:2006}.  Adler, Del\'epine, and
van Moerbeke \cite{Adler:2009a} showed that these distributions also appear 
when 
$n$ non-intersecting Brownian motions start from $x=0$ when $t=0$ with $n-r$
conditioned to end at $x=0$ when $t=1$ and $r$ conditioned to end at
$x=a$ when $t=1$.  The walkers all start out in a single group, but at a
critical time depending on $a$, a group of $r$ walkers separates from the
main bulk.  At this critical time the walkers on the edge where separation
is about to occur follow the \emph{$r$-Airy process}, which is connected with
$F_r(x)$. See \cite{Baik:2006} for other processes in which the functions
$F_r(x)$ arise.

In this paper we extend P\'ech\'e's result in the critical case to more 
general functions 
$V({\bf M})$.  Our specific assumptions are listed in Section \ref{results}, 
but we essentially assume $V({\bf M})$ is a generic analytic potential with 
sufficient growth at infinity.  This establishes a new universality class 
of matrix ensembles with the local eigenvalue density near the critical point 
described by the $r$-Airy process.  The universality of the supercritical 
and subcritical cases has been considered separately \cite{Bertola:2010}.  

In the case of rank one perturbation (i.e. $r=1$), the recent paper \cite{Baik-Wang:2010} by Baik and Wang has described the limiting distribution of the largest eigenvalue for all the possible cases including the critical case that we consider here.  

\subsection{The kernel and its connection to multiple orthogonal polynomials}

Let $p_m(\lambda_1,\dots,\lambda_m)$ be the probability density that
the $n\times n$ matrix ${\bf M}$ chosen using \eqref{prob-measure} has 
eigenvalues
$\{\lambda_1,\dots,\lambda_m\}$ (here $m\leq n$).  Then, when the $\lambda_i$
are distinct, the $m$-point correlation function is
$\frac{n!}{(n-m)!}p_m(\lambda_1,\dots,\lambda_m)$.  Br\'ezin and Hikami
\cite{Brezin:1996, Brezin:1997a, Brezin:1997b, Brezin:1998} showed that
in the Gaussian case, the $m$-point correlation functions can
all be expressed in terms of a single kernel $K(x,y)$:
\eq
\frac{n!}{(n-m)!}p_m(\lambda_1,\dots,\lambda_m) = \det(K(\lambda_i,\lambda_j))_{i,j=1,\dots,m}.
\endeq
Zinn-Justin \cite{Zinn-Justin:1997,Zinn-Justin:1998} extended this result to
the case of more general $V({\bf M})$.  We will find the leading term in 
the large-$n$ asymptotic expansion of the kernel in the critical regime 
near the critical endpoint.

Bleher and Kuijlaars \cite{Bleher:2004a}
showed that the kernel can be written in terms of 
{\it multiple orthogonal polynomials}.  Furthermore, these multiple 
orthogonal polynomials can be written in terms of the solution to a 
certain \emph{Riemann-Hilbert problem}.   Specifically, suppose 
${\bf Y}(z)$ is a $3\times 3$
matrix-valued function of the complex variable $z$ satisfying
\eq
\label{rhp}
\begin{cases}
{\bf Y}(z) \text{ is analytic for } z\notin\mathbb{R}, \\
{\bf Y}_+(x) = {\bf Y}_-(x)\bpm 1 & e^{-nV(x)} & e^{-n(V(x)-ax)} \\ 0 & 1 & 0 \\ 0 & 0 & 1 \epm \text{ for } x\in\mathbb{R}, \\
{\bf Y}(z) = \left({\bf I}+\mathcal{O}\left(\frac{1}{z}\right)\right)\bpm z^n & 0 & 0 \\ 0 & z^{-(n-r)} & 0 \\ 0 & 0 & z^{-r} \epm \text{ as } z\to\infty.
\end{cases}
\endeq
Here ${\bf Y}_\pm(x) := \lim_{\varepsilon\to 0}{\bf Y}(x\pm i\varepsilon)$ denote the
non-tangential limits of ${\bf Y}(z)$ as $z$ approaches the real axis from the
upper and lower half-planes.  Whenever posing a Riemann-Hilbert problem
we assume (unless otherwise stated) that the solution has uniformly 
H\"older continuous boundary
values with any exponent $p\in(0,1]$ along the jump contour when approached 
from either side.
Under our assumption (iv) in Section \ref{results}, the unique 
solution ${\bf Y}(z)$ can be written
explicitly in terms of multiple orthogonal polynomials of the second
kind (see \cite{Bleher:2004b}, Section 2).  In the case of two distinct  
eigenvalues $a$ and $0$, which is our case, 
the kernel $K_n(x,y)$ may be written in terms of the function ${\bf Y}(z)$ as
\eq
\label{mop-kernel}
K_n(x,y) = \frac{e^{-\frac{1}{2}n(V(x)+V(y))}}{2\pi i(x-y)} \bpm 0 & 1 & e^{nay} \epm {\bf Y}(y)^{-1}{\bf Y}(x)\bpm 1 \\ 0 \\ 0 \epm.
\endeq
To analyze the asymptotic behavior of ${\bf Y}$ we will use the standard 
\emph{nonlinear steepest descent method} for Riemann-Hilbert problems, as 
well as certain
ideas introduced by Bertola and Lee \cite{Bertola:2007,Bertola:2008} to 
study the
first finitely many eigenvalues in the birth of a new spectral band for
the random Hermitian matrix model without source.

A potential alternate method for establishing universality for finite 
$r$ would be to use Baik's result \cite{Baik:2008} writing the kernel 
$K_n(x,y)$ in terms of the standard (not multiple) orthogonal
polynomials.  The rank of the matrices in this alternate expression grows 
with $r$, whereas the size of the Riemann-Hilbert problem \eqref{rhp} 
grows with the number of \emph{distinct} eigenvalues.  As such, for 
growing $r$ it is more convenient to analyze the Riemann-Hilbert problem 
for multiple orthogonal polynomials.

\subsection{Definition of the critical regime}

We recall the setting of our work \cite{Bertola:2010}.
Let $g(z)$ be the $g$-function associated with the orthogonal
polynomials with potential $V(z)$ (see, for instance, \cite{Deift:1998-book} 
or \cite{Deift:1999a}).  
It may be written as
\bea
\label{mathfrak-g}
{\mathfrak g}(z;\kappa):=\frac{1}{1-\kappa/2}\int_{\mathbb R}\log(z-s)\rho_{\mbox{\scriptsize min}}(s;\kappa)ds=\log(z)+{\cal O}\left(\frac1z\right),
\eea
where $\rho_\text{min}d s = \rho_\text{min}(s;\kappa) ds$ is the unique measure minimizing the 
functional
\bea
\mathcal{F}[\rho]:=\int_\mathbb{R}V(s)\rho(s)ds-\int_\mathbb{R}\int_\mathbb{R}\rho(s)\rho(s')\log|s-s'|ds\,ds'\ ,\qquad 
\int_\R \rho(s) \d s = 1 - \frac \kappa 2.
\eea 
Here $\kappa$ is a small parameter which is identified with the ratio $\kappa = \frac r n$.
The variational equations are equivalent to the statement that 
\cite{Saff:1997-book} there exists a real constant $\ell_1 = \ell_1(\kappa)$  such that 
\be
\text{Re}\big[V(x) - (2-\kappa) \mathfrak g(x;\kappa)- \ell_1(\kappa)\big]
 \le\{ \begin{array}{cl}
\geq 0 & \text{ on } \R \setminus {\rm supp}\,\rho_\text{min},\\
\equiv 0 & \text{ on } {\rm supp}\,\rho_\text{min}.
\end{array}
\ri.
\label{varprob}
\ee
We shall denote 
\be
g(z):= \mathfrak g(z;0)\ ,\qquad l_1:= \ell_1(0).
\ee
Our first assumption  will be 
\begin{assumption}
The unperturbed ($\kappa=0$) variational  problem is {\bf regular} in the sense of \cite{Kuijlaars:2000} which means that the inequalities in \eqref{varprob} are strict {\em and} the behavior of $V(x) - 2g - l_1$ at any boundary point $x=\xi$  of the support of $\rho$ is asymptotic to $\propto(x-\xi)^\frac 32$ (approaching $\xi$ from the complement of the support). 
\end{assumption}

 It has been shown in \cite{Kuijlaars:2000}  at Theorem 1.3\footnote{In the theorem, the changing parameter is essentially $\kappa$ after rescaling.} that, 
for real-analytic $V(x)$,
\begin{description}
\item[i)] If $V(x)$ is regular for $\kappa=0$ then $V(x)$ is still regular for $\kappa$ small enough, and
\item[ii)] The locations of the spectral edges (the $\alpha$'s and $\beta$'s) 
are real-analytic functions of $\kappa$ such that the bands of the support of 
the equilibrium measure stay separated as $\kappa$ ranges in a small open set around $\kappa=0$.
\end{description}
In addition it is also known that the support (under the real--analyticity assumption) consists of a {\em finite} union of bounded intervals;  we will denote the support of the density $\rho_{\mbox{\scriptsize min}}$   by (see Figure \ref{fig_crit})
\begin{equation}
\begin{split}
{\rm supp}\,\rho_{\mbox{\scriptsize min}}=\left(\bigcup_{j=1}^g[\alpha_j(\kappa),\beta_j(\kappa)]\right)\cup[\alpha(\kappa),\beta(\kappa)],\hspace{.7in}\\
\alpha_1(\kappa)<\beta_1(\kappa)<\alpha_2(\kappa)<\beta_2(\kappa)<\cdots<\alpha_g(\kappa)<\beta_g(\kappa)<\alpha(\kappa)<\beta(\kappa).
\end{split}
\end{equation}

We will consider the {\em unperturbed} density to be the solution of the above variational problem with $\kappa=0$; in this case the reference to $\kappa$ will be tacitly suppressed, and so $\alpha_j= \alpha_j(\kappa)\big|_{\kappa=0}$, etc.
Define for the unperturbed ($\kappa=0$) problem the following quantities:
\begin{eqnarray}
\label{P1-nolog}
P_1(z) & := & -V(z)+2g(z)+l_1, \\
\label{P2-nolog}
P_2(z) & := & -V(z)+az+g(z)+
{l_1-l_3}, \\
\label{P3-nolog}
P_3(z) & := & -P_1(z) + P_2(z) = az-g(z)
{-l_3}.
\end{eqnarray}
Note that $V(\beta) = \frac 1 2 g(\beta)- l_1$ and hence $P_2(\beta) = P_3(\beta)$;  we choose $l_3$ such that $P_2(\beta)=P_3(\beta)=0$.

It is also known that $\text{Re}g(z)$ is a  continuous function on $\R$ and harmonic (and convex) on the complement of the support (up to a sign it  is also known as  the {\em logarithmic potential} in potential theory).

\begin{defn}
\label{acrit}
Define $a_c$ to be the (unique) value of $a$ so that $P_2'(\beta)=0$ (here $\kappa=0$).
\end{defn}
The uniqueness is promptly seen because $P_2'(\beta) = -V'(\beta) + a +g'(\beta)$; 
in fact the effective potential $P_1$  is known \cite{Deift:1998a} to satisfy 
\be
P_1'(z)= \mathcal O{(z-\beta)^\frac 12}.
\ee
In particular, $P_1'(\beta)=0$ and hence $P_2'(\beta) = a-g'(\beta)  = a-\frac 12 V'(\beta)$. Thus the critical value of $a$ is given by
\begin{equation} 
a_c = g'(\beta) = \frac{1}{2} V'(\beta).\label{110}
 \end{equation}

%
%
%
%
We also recall that for  regular potentials the behavior of (a suitable branch of) the function $P_1(z)$ near any of the endpoints  of the interval of support is 
\be
P_1(z) = -C (z-\beta)^{\frac 32} (1+ \mathcal O(z-\beta))
\ee
for some constant $C$.  For the point $\beta = \sup\left[ {\text{supp}} \, \rho_{\text{min}}\right]$ one can also prove that $C>0$; this allows us to introduce the {\em scaling coordinate} $\zeta$ near $z=\beta$ via the definition 
\be
\zeta= \zeta(z;n):= \le(-\frac {3n}4 P_1(z)\ri)^{\frac 23} = n^{\frac 23} c_1 (z-\beta)(1 + \mathcal O(z-\beta))\ ,\ \ c_1>0\label{zetaintro}.
\ee

We now define the critical and near-crtical regimes.  A more extensive context for these definition can be found in \cite{Bertola:2010}.
For completeness we also define the supercritical, subcritical, and 
jumping outlier regimes.  The supercritical and subcritical regimes are 
dealt with separately in \cite{Bertola:2010};  we plan to consider the 
(non-generic) jumping outlier regime in a future work.

\begin{defn}\label{critical}
The matrix model specified by \eqref{prob-measure} is in the {\bf critical regime} if $a=a_c$ and 
{$P_2(x)<P_2(\beta)$ for $x>\beta$.} 
  The scaling regime of $a-a_c  = \mathcal{O}\left(n^{-1/3}\right)$ will be called {\bf near-critical}.  We define the exploration parameter $\tau$ by
\bea
\tau:=\lim_{n\to\infty} n^{1/3}(a-a_c)/c_1,
\eea
where $c_1$ is the positive constant defined at 
\eqref{zetaintro}.
\end{defn}
We define, for $0<a<a_c$, $b^*$ to be the unique point on the real axis 
greater than $\beta$ such that $P_3'(\beta)=0$.  For $a\geq a_c$ we can 
choose $b^*:=\beta$.
\begin{defn}
\label{supercritical}
The model is in the {\bf supercritical regime} if $P_2$ has a unique global maximum on $\{x\geq \max\{\beta, b^*\}\}$ at a point $x=a^*\in \R$ and any of the three conditions below is satisfied:
\begin{itemize}
\item $a>a_c$.
\item $a=a_c$ and $P_3(\beta) =P_2(\beta)<P_2(x)$ for some $x>\beta$.
\item $0<a< a_c$ and $P_3(b^{*})<P_2(x)$ for some $x>b^{*}$.
\end{itemize}
Note that $a^*$ is always greater than $\beta$ and $b^*$.
If the global maximum on $[\max\{\beta,b^*\}, \infty)$ is  attained at several distinct points then we will say that we are in the {\bf jumping outlier regime}.
\end{defn}

\begin{defn}\label{subcritical}
The matrix model specified by \eqref{prob-measure} is in the {\bf subcritical regime} if $a < a_c$ and 
$P_2(x)<P_3(b^*)$ for all $x\geq b^*$.
\end{defn}


\subsection{Assumptions and results}
\label{results}

We will make the following assumptions on $a$, ${\bf A}$, and $V(z)$:  
\begin{itemize}
\item[(i)]  $a>0$.
\item[(ii)]  ${\bf A}$ is a \emph{small-rank} external source of the form 
\eqref{A} with $r=\mathcal{O}(n^\gamma)$, $0\leq\gamma<1/12$.  
\item[(iii)]  $V(z)$ is real analytic and \emph{regular} in the sense of 
\cite{Kuijlaars:2000}.
\item[(iv)] $\ds\lim_{|z|\to\infty}\frac{V(z)}{\log(1+z^2)} = \infty \quad \text{and} \quad \lim_{|z|\to\infty}\frac{V(z)-az}{\log(1+z^2)} = \infty.$
\end{itemize}

Regarding assumption (i), the case when $a<0$ is equivalent by 
sending $a\to-a$
 and $V(z)\to V(-z)$.  As for assumption (ii), 
in the general case when ${\bf A}$ has $m>2$ distinct eigenvalues the 
kernel can be written in terms of multiple orthogonal polynomials associated 
to an $(m+1)\times(m+1)$ Riemann-Hilbert problem, which is beyond the scope of 
this paper.  The assumption of analyticity in (iii) allows us 
to use the nonlinear steepest-descent method for Riemann-Hilbert problems.  
The assumption of regularity ensures that the equilibrium measure of $V(z)$ has square-root 
decay at each band endpoint (so that  we can use Airy parametrices) and that these 
endpoints are analytic functions of $\kappa$ 
{near $\kappa=0$.}
 
Next, (iv) guarantees the existence of the multiple orthogonal polynomials 
needed to ensure the Riemann-Hilbert problem has a solution.  
We note that the allowed $V(z)$ include 
any convex $V(z)$ (see the introduction of \cite{Bertola:2010}).  

We compute the large-$n$ behavior of the kernel function \eqref{mop-kernel} 
in the critical regime.  We explicitly compute the kernel in a neighborhood of 
$\beta$.  In the remaining portions of the complex
plane, our result is that the kernel function converges to the kernel for the 
classical orthogonal polynomial problem with respect to $V(x)$.  That is, 
away from $\beta$ the standard universality classes apply 
(i.e. the sine kernel in the bulk of the spectrum and Airy kernels at the 
other edges).  Our main result is:
\begin{shaded}
\begin{thm} \label{crit-kernel-thm}
Suppose $V(z)$ and $a$ satisfy conditions (i)--(iv).  Let $\zeta_x, \zeta_y$ be fixed in some bounded set.
Also let $c_1$ be the constant appearing in 
\eqref{zetaintro}.
Then for large $n$, and for a positive integer $r$ such 
that $r=\mathcal{O}(n^\gamma)$ with $0\leq\gamma<1/12$, 
\eq
K_n\le(
\beta + \frac {\zeta_x+\delta}{c_1 n^{\frac 23}}, \beta + \frac {\zeta_y+\delta}{c_1 n^{\frac 23}}
\ri)
 = \frac{c_{1}n^{2/3}}{(2\pi i)^2}
\int_{ \widetilde{\mathcal{C}} } ds \int_{\mathcal{C}}  dt
\frac{ (t+\tau)^r}{(s+\tau)^r} \frac{ e^{\frac{1}{3} (s^3 - t^3) + \zeta_x t - \zeta_y s }}{ t- s}\left(1+{\cal O}\left(\frac{1}{n^{1/3}}\right)\right),
\label{Knkern}
\endeq
where the oriented contours $\mathcal{C}$ and $\widetilde{\mathcal{C}}$ are given in Figure \ref{AiryKernelContours}
and $\delta$ is a quantity independent of $\zeta_x,\zeta_y$ and of the form 
\bea
\delta := c_1 \dot \beta \kappa n^{\frac 2 3} = \mathcal O(n^{\gamma - \frac 1 3})\ ,\qquad \dot\beta := \frac {\d \beta(\kappa)}{\d \kappa} \bigg|_{\kappa=0}\ .
\eea
\end{thm}
\end{shaded}

\begin{remark}
By dropping the drift term $\delta$ in \eqref{Knkern} we would simply deteriorate the error estimate to $\mathcal O(n^{\gamma-\frac 13})$ which is however  still vanishing since $\gamma<\frac 1 {12}$.
\end{remark}
The constant $\dot \beta$ admits an explicit integral representation for an arbitrary real-analytic potential $V(z)$ but it is a bit complicated when the 
equilibrium measure is supported on multiple intervals. In the simplest case where the support of the equilibrium measure consists of a single interval $[\alpha,\beta]$ then we have 
\be
\dot \beta =\frac 1{(\beta -\alpha)\oint \frac{V'(z)\d z}{(z-\beta) R(z) 2i\pi}},
\ee
where the contour of integration is a simple closed contour surrounding the support $[\alpha,\beta]$ in the complex plane. We will not be using in any way the explicit form of $\dot \beta$, except the fact that it is a well--defined quantity due to the smoothness of $\beta(\kappa)$ guaranteed by the already cited  Kuijlaars' Theorem 1.3 in \cite{Kuijlaars:2000}.

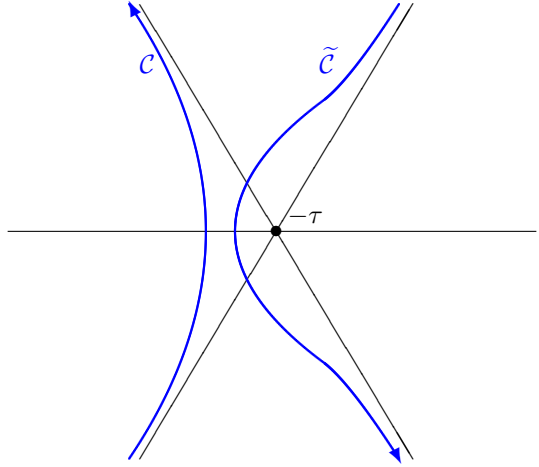
\begin{figure}
\setlength{\unitlength}{2pt}
\begin{center}
\begin{picture}(100,100)(-50,-50)
\put(-50,0){\line(1,0){100}}
\put(-25,43){\line(3,-5){51.7}}
\put(-25,-43){\line(3,5){51.7}}
\put(0.8,0){\circle*{2}}
\put(3,1.5){$-\tau$}
\thicklines
\blue{
\put(-25,30){$\mathcal{C}$}
\qbezier(-27,43)(2,0)(-27,-43)
\put(-26.1,41.6){\vector(-1,2){1}}
\put(9,30){$\widetilde{\mathcal{C}}$}
\qbezier(10,25)(-24,0)(10,-25)
\qbezier(10,25)(15,29)(24,43)
\qbezier(10,-25)(15,-29)(24,-43)
\put(23.5,-42){\vector(1,-2){1}}
}
\end{picture}
\end{center}
\caption{The contours $\mathcal{C}$ and $\widetilde{\mathcal{C}}$ used in the $r$-Airy Kernel.  The straight lines cross at $-\tau$.\label{AiryKernelContours}
}
\end{figure}
We show in Section \ref{brownian-motions} that our kernel is the same as the 
one found for nonintersecting Brownian walkers in \cite{Adler:2009a}.  


\bigskip
\noindent
{\bf Acknowledgments.}  The authors would like to thank Jinho Baik,
Ken McLaughlin, and Dong Wang for several illuminating discussions.  M. 
Bertola 
was supported by NSERC.  R. Buckingham was 
supported by the Charles Phelps Taft Research Foundation.  
V. Pierce was supported by 
NSF grant DMS-0806219.

\section{The perturbed equilibrium measure and the local coordinate \texorpdfstring{$\bs{\zeta}$}{zeta}} 
To describe a growing number of outliers, we define a perturbed 
 equilibrium measure problem with parameter of perturbation $\kappa = r/n= {\cal O}(n^{\gamma-1})$ for $0\leq\gamma<1/12$.
Using the $\mathfrak g$-function \eqref{mathfrak-g} we define
\begin{shaded}
\begin{equation}\label{141}
\begin{split}
{\mathcal P}_1(x;\kappa) & :=  -V(x) + (2-\kappa){\mathfrak g}(x;\kappa)  + \ell_1, \\ 
{\mathcal P}_2(x;\kappa) & := -V(x) + ax +\frac12(2-\kappa){\mathfrak g}(x;\kappa) + \frac{3\kappa}2\log(x-\beta) + \ell_1 - \ell_3 , \\ 
{\mathcal P}_3(x;\kappa) & :=
{-}{\mathcal P}_1 
{+} {\mathcal P}_2 = 
ax 
{-}\frac12(2-\kappa){\mathfrak g}(x;\kappa) 
{+} \frac{3\kappa}2\log(x-\beta) 
{-} \ell_3.
\end{split}
\end{equation}
\end{shaded}
 The other constant 
$\ell_3$ will be defined by \eqref{elltilde} after we define 
the locally holomorphic function $h(z)$.  

We also note that   we have
\begin{equation}
\lim_{\kappa\rightarrow0}{\cal P}'_j(x)=P'_j(x),\quad j=1,2,3,
\end{equation}
uniformly over compact subsets of $\R\setminus \{\beta\}$.
The $P_j(x)$'s ($j=1,2,3$) are defined by the same equations as the ${\cal P}_j$'s except that $\kappa$ is set to zero (hence all the log terms are dropped) and that ${\mathfrak g}(x)$ is replaced by $g(x)=\lim_{\kappa\rightarrow 0}{\mathfrak g}(x)$.  Such convergence is uniform (outside a finite disk around $\beta$ and inside a compact set) according to \cite{Kuijlaars:2000}. 

Since $V$ is regular for $\kappa=0$, it will still remain regular for small values of $\kappa$.  Hence
there exists a holomorphic function 
$\zeta(z)$ in a finite disk around $\beta$ such that
\begin{equation}\label{zeta}
-n{\cal P}_1(z;\kappa)=\frac43\zeta(z;\kappa)^{3/2}, \quad \zeta(\beta(\kappa);\kappa)=0,
\end{equation}
and thus
\eq
\label{zeta-ito-zmbeta}
\zeta = n^{2/3}c_1(\kappa)(z-\beta(\kappa))(1 + O(z-\beta(\kappa)))
\endeq
for some constant $c_1(\kappa)>0$.

Previously, we have defined $a_c$ at $\kappa=0$.  To describe the effect 
of growing $r$ we introduce a more exact definition $a_c(\kappa)$.  Note 
that 
\begin{equation}\label{hdef}
h(z;\kappa):=\frac12{\cal P}_1(z;\kappa)
{+}{\cal P}_3(z
;\kappa)-\frac{3\kappa}{2}\log\zeta
=-\frac12 V(z)+az+\frac{3\kappa}{2}\log\frac{z-\beta}\zeta  + \frac{\ell_1(\kappa)
}{2} - \ell_3(\kappa)
\end{equation}
is a locally holomorphic function at $z=\beta$ because $V(z)$ is real analytic.
\begin{defn}[$a_c(\kappa)$ and $\ell_3(\kappa)$]\label{ack}
For $\kappa>0$, $a_c(\kappa)$ is such that $h'(\beta(\kappa);\kappa)=0$ at $a=a_c(\kappa)$, i.e.
\begin{equation}
\label{acrit-kappa}
a_c(\kappa):=\frac12V'(\beta(\kappa))-\frac{3\kappa}{2}\frac{d}{dz}\log\frac{z-\beta
(\kappa) }\zeta\Big|_{z=\beta(\kappa)}.
\end{equation}
We also define $\ell_3(\kappa)$ such that 
\begin{equation}
h(\beta(\kappa);\kappa)=\frac{1}{n}\left(-\log\Gamma(r) - \log\left(\frac{\eta_0(\tau)}{2\sqrt\pi i}\right)\right),
\label{elltilde}
\end{equation}
where $\Gamma(r)$ is Euler's Gamma function and 
\begin{equation} \label{eta-naught}
\eta_0(\tau) := \int_0^\infty Ai(t) \left( e^{\tau t/\omega} + e^{\omega \tau t} + e^{\tau t} \right) dt\, ,\qquad
\omega := \exp\left( \frac{2\pi i}{3} \right) \,.
\end{equation}

One finds explicitly 
\begin{equation}
\ell_3(\kappa):= \frac{1}{n} \left( \log \Gamma(r) + \log\left( \frac{\eta_0(\tau)}{2 \sqrt{\pi} i} \right) \right) 
- \frac{1}{2} V(\beta(\kappa)) + a \beta
(\kappa)
 - \frac{\ell_1
(\kappa)
}{2} + \frac{3\kappa}{2} \log\left( \frac{1}{n^{2/3} c_1
(\kappa)
} \right)\,. 
\end{equation}
\end{defn}
We observe that if $r$ grows with $n$ then $\text{Re}(h(\beta
(\kappa);\kappa
))<0$ for 
sufficiently large $n$, which will be used in the proof of Proposition \ref{crit-lem-1}.
One can verify that $a_c(\kappa)$ converges to $a_c:= a_c(0)$ (see Definition 
\eqref{acrit}) when $\kappa\rightarrow0$
at the rate 
\be
a_c(\kappa) = a_c + \mathcal O(\kappa ) =  a_c + {\cal O}(n^{\gamma-1})\ .
\ee
Since for $\kappa=0$ we have $a_c=\frac 12 V'(\beta)=g'(\beta)>0$,  for 
sufficiently small $\kappa$ we still  have $a_c(\kappa)>0$.


 From Definition \ref{ack}, we have $h(z)=h(\beta
(\kappa);\kappa )+{\cal O}((z-\beta
 (\kappa)
)^2)$ at $a=a_c(\kappa)$ (because $h'(\beta(\kappa);\kappa
)=0$).  For other values of $a$, since only the term $az$ in $h(z
;\kappa
)$ depends on $a$, we get $h(z
;\kappa)={\cal O}((z-\beta
(\kappa))^2)+(a-a_c(\kappa))z$.  

For the subsequent exposition, we redefine $\tau$ in a way that is compatible with the earlier Definition \ref{critical}.
\begin{defn}[replacing Definition \ref{critical}]
\begin{equation}
\tau := n^{1/3}\left(a-a_c(\kappa)\right)/c_1(\kappa).
\end{equation}
\end{defn}
From \eqref{zeta-ito-zmbeta} and the above definition, we get
\begin{equation} \label{hz1}
h(z;\kappa)=h(\beta(\kappa);\kappa
)+\frac{\tau}{n}\zeta+{\cal O}\left((z-\beta
(\kappa)
)^2\right) = h(\beta
(\kappa);\kappa
)+\frac{c_1
(\kappa)
}{n^{1/3}}\tau(z-\beta
(\kappa)
)+{\cal O}\left((z-\beta
(\kappa)
)^2\right).
\end{equation}

\section{Initial analysis of the Riemann-Hilbert problem:  the global parametrix}
\label{Sec_lens}

We define the contour $L$ as the positively-oriented circle centered at $\alpha_1$ (the leftmost edge of the spectrum) and passing through $\beta$ (the rightmost edge of the spectrum). We choose the circle $L$ large enough so that ${\cal P}_2$ is negative on the real axis to the left of $L$.
Until further notice the dependence of the various quantities on $\kappa$ 
(i.e. $a_c$, $\beta$, $\mathfrak{g}$, etc.) will be understood throughout.
We have:
\begin{lemma}\label{lem_steepest}
For  sufficiently small $\kappa$ and $a=a_c$, the function ${\rm Re}\!\left(
\insrt{-}
{\cal P}_3(z)+\frac{3\kappa}{2}\log(z-\beta)\right)$ increases as one follows $L$  in either direction starting from $\beta$ (i.e. through the upper half-plane or the lower half-plane).
\end{lemma}
\begin{proof}
From the definition \eqref{141}, we have $\insrt{-}{\cal P}_3(z)+\frac{3\kappa}{2}\log(z-\beta)=-az +\frac12(2-\kappa){\mathfrak g}(z) - \widetilde\ell$. It is obvious that ${\rm Re}(-az)$ increases when $a=a_c>0$. It is also simple to see that ${\rm Re}\left(\frac12(2-\kappa){\mathfrak g}(z)\right)$ increases along the referred contour because ${\rm Re}\log(z-x)$ increases along the contour for any $x$ in $[\alpha_1,\beta]$.
\end{proof}
We now open up lenses around each of the bands in the standard way (as in the 
analysis of the orthogonal polynomials associated to $V(z)$).  We introduce 
the following open regions (see Figure \ref{fig_crit}):
\begin{description}
\item $\Omega_j^\pm$, $j=1,...,g$:  The area in the upper half-plane ($+$) 
or lower half-plane ($-$) between the band $[\alpha_j,\beta_j]$ and its 
appropriate adjacent lens.
\item $\Omega_\text{main}^\pm$:  The area in the upper half-plane ($+$) 
or lower half-plane ($-$) between the band $[\alpha,\beta]$ and the 
appropriate adjacent lens.
\item $\Omega_L^\pm$:  The area in the upper half-plane ($+$) 
or lower half-plane ($-$) inside the contour $L$ but outside the lenses.
\item $\Omega_\text{out}^\pm$:  The part of the upper half-plane ($+$) or 
lower half-plane ($-$) outside the contour $L$.  
\end{description}
We also define
\eq
\Omega_\text{lens}^\pm:=\left(\bigcup_{j=1}^g\Omega_j^\pm\right)\cup\Omega_\text{main}^\pm, \quad \Omega_\text{lens}:=\Omega_\text{lens}^+\cup\Omega_\text{lens}^-, \quad \Omega_L:=\Omega_L^+\cup\Omega_L^-, \quad \Omega_\text{out}:=\Omega_\text{out}^+\cup\Omega_\text{out}^-. 
\endeq

\begin{figure}
\begin{center}\resizebox{1\textwidth}{!}{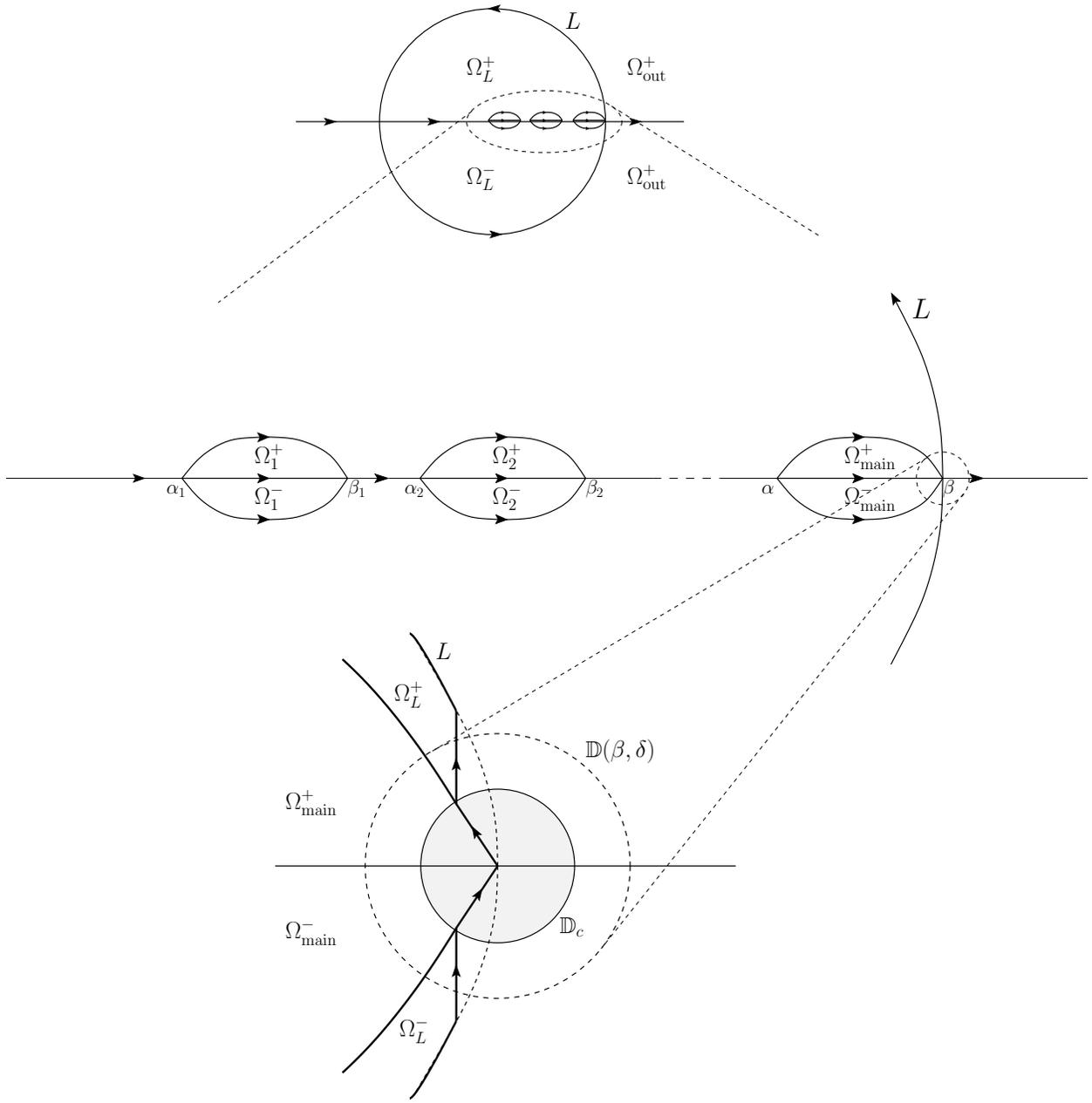}
\caption{The contours for the critical case.  When we construct the 
local parametrix near $\beta$ we will deform the contour $L$ inside 
$\mathbb{D}_c$ such that $\partial\Omega_\text{main}^\pm$ and $L$ overlap.}\label{fig_crit}
\end{center}
\end{figure}

With these definitions of regions and contours, we define ${\bf W}(z)$ in each region by
\eq
\label{W}
{\bf W}(z):= {\bf \Lambda}{\bf Y}(z) \bpm e^{-\frac{n}{2}V} & 0 & 0 \\0 & e^{\frac{n}{2}V} & 0 \\0 & 0 & e^{\frac{n}{2}(V-2a z)} \epm {\bf J}(z) \bpm e^{-\frac{n}{2}{\cal P}_1} & 0 & 0 \\0 & e^{\frac{n}{2}{\cal P}_1} & 0 \\0 & 0 & e^{\frac{n}{2}({\cal P}_1-2{\cal P}_3)} \epm e^{-\frac{n}{2}\kappa\log(z-\beta)},
\endeq
where
\eq
{\bf \Lambda}:= \bpm e^{\frac{n}{2}\ell_1} & 0 & 0 \\0 & e^{-\frac{n}{2}\ell_1} & 0 \\0 & 0 & e^{\frac{n}{2}(2\ell_3 - \ell_1)} \epm
\endeq
and
\eq
\begin{array}{lcccr} {\bf J}(z):= \mbox{\vspace{2in}}  & \left\{ \quad \begin{array}{c} \\ {\bf I}, \\ {} \end{array} \quad \right.  & \bpm 1 & 0 & 0 \\ 0 & 1 & -1 \\ 0 & 0 & 1 \epm, & \left. \bpm 1 & 0 & 0 \\ \mp 1 & 1 & -1 \\ 0 & 0 & 1 \epm  \right\}. &  \\ & \hspace{.2in}z\in\Omega_\text{out} & z\in\Omega_L & z\in\Omega_\text{lens}^\pm &  \end{array} 
\endeq
Then ${\bf W}(z)$ satisfies the following jump conditions (see Figure 
\ref{fig_crit} for the orientation of the contours):
\begin{equation}\label{Ajumpcrit}
{\bf W}_+(z)={\bf W}_-(z) \times \left\{\begin{array}{ll}\vspace{0.1cm}
\bpm 1&0&0\\0&1&-e^{
n{\cal P}_3(z)}\\0&0&1\epm,& z\in L=\partial\Omega_\text{out}\cap\partial\Omega_L,
\\\vspace{0.1cm}\bpm 1 & 0 & 0 \\ e^{-n{\cal P}_1(z)} & 1 & 0 \\ 0 & 0 & 1 \epm,&z\in \partial\Omega_\text{lens}\cap\partial\Omega_L,
\\\vspace{0.1cm}\bpm 0&(-1)^{r}&0\\(-1)^{r+1}&0&0\\0&0&1\epm,&z\in\left(\bigcup_{j=1}^g[\alpha_j,\beta_j]\right)\cup[\alpha,\beta],
\\\vspace{0.1cm}\bpm e^{- i n\sigma(z)}(-1)^{r}&(-1)^{r}e^{n{\rm Re}{\cal P}_1(z)}&0\\0&(-1)^{r}e^{ i n\sigma(z)}&0\\0&0&1\epm,&z\in
\partial\Omega_L^+\cap\partial\Omega_L^-,
\\\vspace{0.1cm}\bpm 1 & e^{n{\cal P}_1(z)} & e^{n{\cal P}_2(z)} \\ 0 & 1 & 0 \\ 0 & 0 & 1 \epm,& z\in \partial\Omega_\text{out}^+\cap\partial\Omega_\text{out}^-,
\end{array}\right.
\end{equation}
and the boundary condition
\begin{equation}\label{Ajumpcrit2}
{\bf W}(z)={\bf I}+{\cal O}\left(\displaystyle\frac{1}{z}\right) ,\quad z\to\infty.
\end{equation}
Here we have defined $\sigma(x)$ by $\sigma(x):=\frac{1}{1-\frac{\kappa}{2}}\int_{-\infty}^x \rho_{\mbox{\scriptsize min}}(s;\kappa)ds$ which is some constant on each gap.

We will show below that the above Riemann-Hilbert problem is exponentially 
close to a simpler one away from the turning points such 
as $\beta$.
To be more precise, let us define a shrinking disk centered at $\beta$ by
\begin{equation}\label{Dc-def}
{\mathbb D}_c:=\left\{z\,\big|\,|z-\beta|<n^{-1/2}\right\}.
\end{equation}
Here we have chosen the diameter of the disk ${\mathbb D}_c$ so that (refer to \eqref{hz1})
\begin{equation}
\label{h-behavior}
n\left(h(z)-h(\beta)-\frac{\tau}{n}\zeta \right)=n\,{\cal O}\left((z-\beta)^2\right)
\end{equation}
 is uniformly bounded on ${\mathbb D}_c$ as $n\rightarrow\infty$.
From the expansion \eqref{zeta-ito-zmbeta} and \eqref{Dc-def} we note
\eq
\label{zeta-N-exp}
\zeta = \mathcal{O}(n^{1/6}) \text{ on } \partial{\mathbb D}_c.
\endeq
To complete the error analysis we will also need a fixed-size disk 
around $\beta$.  Fix a small $\delta>0$ and define 
${\mathbb D}(\beta,\delta)$ to be the finite disk centered at $\beta$ with 
fixed radius $\delta$.  For $n$ large enough, we have 
${\mathbb D}_c\subset {\mathbb D}(\beta,\delta)$.
We will show in Section \ref{Error-Section} that, for $z$ outside $\mathbb{D}_c$ and 
a finite distance away from the other turning points (the $\alpha_j$'s, 
$\beta_j$'s, and $\alpha$), 
the Riemann-Hilbert problem for ${\bf W}(z)$ is exponentially close to that 
for the outer-parametrix ${\bf \Psi}$ as $n\to \infty$ with the proper choices of $\ell_1$ and $\ell_3$.  The necessary data on the effective potentials $\mathcal{P}_1(z)$, $\mathcal{P}_2(z)$, and $\mathcal{P}_3(z)$ is contained in Lemma \ref{crit-lem-1}.  
Thus we propose the following model Riemann-Hilbert problem for the 
{\it outer parametrix}.  
\begin{defn}
In the critical case, the outer parametrix ${\bf \Psi}(z)$ is defined as the solution of the 
Riemann-Hilbert problem
\begin{equation}\label{model_crit}
\left\{\begin{array}{ll}\vspace{0.2cm} {\bf \Psi}_+(z)={\bf \Psi}_-(z)\bpm 0&(-1)^r &0\\(-1)^{r+1}&0&0\\0&0&1\epm,&z\in\left(\bigcup_{j=1}^g[\alpha_j,\beta_j]\right)\cup[\alpha,\beta],
\\\vspace{0.2cm} {\bf \Psi}_+(z)={\bf \Psi}_-(z)\bpm (-1)^re^{-i n\sigma(z)}&0&0\\0&(-1)^r e^{i n\sigma(z)}&0\\0&0&1\epm,&z\in
{\mathbb R}\setminus \left(\left(\bigcup_{j=1}^g[\alpha_j,\beta_j]\right)\cup[\alpha,\beta]\right) ,
\\\vspace{0.2cm}{\bf \Psi}(z)={\cal O}\left((z-\star)^{-1/4}\right), &z\rightarrow\star, \quad\star=\alpha_j, \beta_j,\alpha,\beta.
\\\vspace{0.2cm}{\bf \Psi}(z)= {\bf I}+{\cal O}\left(\frac{1}{z}\right),~& z\to\infty.
\end{array}\right.
\end{equation}
The function $\sigma(z)$ is defined below the equation \eqref{Ajumpcrit2}.
\end{defn}
This Riemann-Hilbert problem is essentially $2\times2$. 
The unique solution is given in \cite{Deift:1999a} Lemma 4.3, or in \cite{Bertola:2009b} Lemma 4.6 in a slightly more generalized setup.  We refer the interested reader to those papers since this is not essential here.  In the subsequent analysis we will only need the following 
information.
\begin{lemma}
\label{H-lemma}
Define 
\eq
\label{S-of-zeta}
{\bf S}(\zeta;n):= \bpm \left(\frac{\zeta}{n^{2/3}}\right)^{-1/4} & 0 & 0 \\ 0 & \left(\frac{\zeta}{n^{2/3}}\right)^{1/4} & 0 \\ 0 & 0 & 1 \epm {\bf U} \bpm (-1)^r & 0 & 0 \\ 0 & 1 & 0 \\ 0 & 0 & (-1)^r \epm; \quad {\bf U}:=\bpm \frac{1}{2} & \frac{i}{2} & 0 \\ -\frac{1}{2} & \frac{i}{2} & 0 \\ 0 & 0 & -2i \epm.
\endeq
Then for $z\in\mathbb{D}_c$, there is a unique holomorphic
$3\times 3$ matrix-valued function ${\bf H}_{(0)}(z;n)$ with determinant 
one such that
\eq
\label{Psi-ito-H}
{\bf \Psi}(z) = {\bf H}_{(0)}(z) {\bf S}(\zeta)
\endeq
as $\zeta\to\infty$.  In addition, ${\bf H}_{(0)}(z;n)$ has a limit 
as $n\to\infty$ and 
\eq
\label{H0-limit}
{\bf H}_{(0)}(z;n) - \lim_{n\to\infty}{\bf H}_{(0)}(z;n) = \mathcal{O}(\kappa).
\endeq
\end{lemma}
\begin{proof}
A direct check shows that 
${\bf S}(\zeta;n)$ has the same jumps as ${\bf \Psi}(z)$ 
in a neighborhood of $z=\beta$. 
Thus the ratio ${\bf \Psi S}^{-1}$ has no jump discontinuities inside 
$\mathbb{D}_c$ and hence may have {\em at most} an isolated singularity at $z=\beta$. Furthermore, this product has at worst a square-root 
singularity at $\beta$ (coming from the product of the quarter-root 
singularities in ${\bf \Psi}$ and ${\bf S}$).  In the absence of a 
branch cut, this means ${\bf H}_{(0)}$ is holomorphic.  Finally, 
\eqref{H0-limit} follows from the definition of $\zeta(z)$ in \eqref{zeta} 
and the dependence of $\mathcal{P}_1$ on $\kappa$.  
\end{proof}

\section{The local parametrix near \texorpdfstring{$\bs{z=\beta}$}{sdsa}}

We begin this section by expressing the Riemann-Hilbert problem satisfied by ${\bf W}$ {\em inside} ${\mathbb D}_c$ in terms of the local coordinate $\zeta$.
Zooming in on ${\mathbb D}_c$, the contours are shown in Figure \ref{fig_critloc}.  
There we collapse the global contours $L$ and a part of $\partial\Omega_{\text{main}}$ into $\Gamma_2$ 
and $\Gamma_4$.  The regions II and III are parts of the region $\Omega_{\text{main}}$.

Using the identities below,
\begin{align}
{\cal P}_2(z)=\frac12{\cal P}_1(z)+h(z)+\frac{3\kappa}{2}\log\zeta,
\\{\cal P}_3(z)=
{-}\frac12{\cal P}_1(z)
{+}h(z)
{+}\frac{3\kappa}{2}\log\zeta,
\end{align}
we see that ${\bf W}(z)$ satisfies the following Riemann-Hilbert problem inside the disk ${\mathbb D}_c$:
\begin{equation}\label{Ajumplocal}
{\bf W}_+(z)={\bf W}_-(z)\left\{\begin{array}{ll}\vspace{0.2cm}
\bpm 1 & e^{-\frac43\zeta^{3/2}} & \zeta^{\frac{3r}{2}}e^{-\frac23\zeta^{3/2}+nh(z)}\\ 0 & 1 & 0 \\ 0 & 0 & 1 \epm,& z\in \Gamma_1,
\\\vspace{0.2cm} 
\bpm 1&0&0\\-e^{\frac43\zeta^{3/2}}&1&-\zeta^{\frac{3r}{2}}e^{\frac23\zeta^{3/2}+nh(z)}\\0&0&1\epm,& z\in \Gamma_2,
\\\vspace{0.2cm} 
\bpm 0&(-1)^r&0\\(-1)^{r+1}&0&0\\0&0&1\epm,&z\in\Gamma_3,
\\\vspace{0.2cm}
\bpm 1&0&0\\e^{\frac43\zeta^{3/2}}&1&-\zeta^{\frac{3r}{2}}e^{\frac23\zeta^{3/2}+nh(z)}\\0&0&1\epm,& z\in\Gamma_4.
\end{array}\right.
\end{equation}
We construct below a local parametrix that solves a Riemann-Hilbert 
problem similar to the above.

\subsection{\texorpdfstring{$\bs{r}$}{bsr}-Airy parametrix}


\begin{figure}
\setlength{\unitlength}{1.7pt}
\begin{center}
\begin{picture}(220,100)(-150,-50)
\put(0.8,0){\circle*{2}}
\put(3,1.5){$-\tau$}
\blue{
\thicklines
\put(-20,-2){$\mathcal{C}_1$}
\qbezier(-29,42)(2,0)(-29,-42)
\put(-28.5,41.6){\vector(-1,2){1}}
\put(12,15){$\mathcal{C}_2$}
\qbezier(-21,42)(3,3)(50,4)
\put(-20.5,41.6){\vector(-1,2){1}}
\put(12,-17){$\mathcal{C}_6$}
\qbezier(-21,-42)(3,-3)(50,-4)
\put(-20.5,-41.6){\vector(-1,-2){1}}
\put(53,-2){$\mathcal{C}_4$}
\put(0.8,0){\vector(1,0){50}}
\put(-27,-48){$\mathcal{C}_5$}
\put(-25,-43){\line(3,5){25.4}}
\put(-24.1,-41.6){\vector(-1,-2){1}}
\put(-28,45){$\mathcal{C}_3$}
\put(-25,43){\line(3,-5){25.4}}
\put(-23.8,41.1){\vector(-1,2){1}}
}
\put(-94.8,0){\circle*{2}}
\put(-85,15){$\rm{\bf I}$}
\put(-125,15){$\rm{\bf II}$}
\put(-126,-16){$\rm{\bf III}$}
\put(-86,-16){$\rm{\bf IV}$}
\thicklines
\put(-278,-2){$\Gamma_1$}
\put(-135,0){\line(1,0){80}}
\put(-75,0){\vector(1,0){1}}
\put(-123,45){$\Gamma_2$}
\put(-120,43){\line(3,-5){25.4}}
\put(-107,21.8){\vector(-1,2){1}}
\put(-142,-2){$\Gamma_3$}
\put(-115,0){\vector(1,0){1}}
\put(-123,-49){$\Gamma_4$}
\put(-120,-43){\line(3,5){25.4}}
\put(-107,-21.8){\vector(1,2){1}}
\end{picture}
\end{center}
\caption{Contours where ${\cal A}_r$ has jumps.  The location of the center can be in any finite domain.}\label{fig_critloc}
\vspace{-0.2cm}
\caption{Contours to define the generalized Airy functions. The center is located at $-\tau$.}\label{fig_airycontour}
\end{figure}

The $r$th derivative of the standard Airy function admits the contour integral representation
\begin{equation}
\frac{d^r}{d\zeta^r}{\rm Ai}(\zeta)=\frac{1}{2i\pi}\int_{{\cal C}_1}t^r e^{\zeta t-t^3/3}dt\  , \quad r=0,1,2,...,
\end{equation}
where the contour ${\cal C}_1$ is shown in Figure \ref{fig_airycontour}.
Extending the standard Airy function, we will need the following generalized Airy functions.

\begin{defn}\label{155}
Let us define the following generalized Airy functions corresponding to each contour in Figure \ref{fig_airycontour}.
\begin{equation}\label{AiCdef}
{\rm Ai}_{{\cal C}_m}^{(r)}(\zeta;\tau):=\frac{1}{2i\pi}\int_{{\cal C}_m}(t+\tau)^r e^{\zeta t-t^3/3}dt ,
\end{equation}
where $m=1,...,6$ indicates the contour depicted at Figure \ref{fig_airycontour}, and $r$ is a non-negative integer.
\end{defn}


Using the generalized Airy functions, we shall construct a
 matrix ${\cal A}_r(\zeta)$ satisfying the jump condition (see Figure \ref{fig_critloc}) below.
\begin{equation}\label{ARHP}
\begin{array}{c}
{}\vspace{0.2cm}\\
{\cal A}_{r}(\zeta)_+={\cal A}_{r}(\zeta)_-\times
\end{array}
\hspace{-0.2cm}
\begin{array}{cccc}
\zeta\in\Gamma_1&\zeta\in\Gamma_2&\zeta\in\Gamma_3&\zeta\in\Gamma_4 \vspace{0.2cm}\\
\left\{\bpm1&1&1\\0&1&0\\0&0&1\epm,\right. &
\bpm 1&0&0\\-1&1&-1\\0&0&1\epm, &
\bpm 0&1&0\\-1&0&0\\0&0&1\epm, &
\left.\bpm 1&0&0\\1&1&-1\\0&0&1\epm\right\}.
\end{array}
\end{equation}
\begin{thm}
For a positive integer $r$, the following definition satisfies the jump condition in \eqref{ARHP}.
\begin{equation}
{\cal A}_r(\zeta):= \bpm v^{(r)}_1(\zeta) & v^{(r)}_2(\zeta)& v^{(r)}_3(\zeta)
\\\partial_\zeta v^{(r)}_1(\zeta) &\partial_\zeta v^{(r)}_2(\zeta)&\partial_\zeta v^{(r)}_3(\zeta)
 \\v^{(r-1)}_1(\zeta) &v^{(r-1)}_2(\zeta)&v^{(r-1)}_3(\zeta)\epm
\left(\begin{array}{ccc}1 & 0 & 0 \\\bigstar & 1 & 0 \\0 & 0 & 1\end{array}\right),\quad \bigstar:=\begin{cases}-1\quad\mbox{for II},\\1\quad\mbox{for III},\\0\quad\mbox{for I and IV},\end{cases}
\label{47}
\end{equation}
where
\begin{equation}
v_1^{(r)}(\zeta)={\rm Ai}^{(r)}_{{\cal C}_1}(\zeta),\quad v_2^{(r)}(\zeta)=\begin{cases}{\rm Ai}^{(r)}_{{\cal C}_2}(\zeta), ~~\zeta\in\mbox{I and II},\vspace{0.1cm}
\\{\rm Ai}^{(r)}_{{\cal C}_6}(\zeta),~~\zeta\in \mbox{III and IV},\end{cases}\quad
v^{(r)}_3(\zeta)=\begin{cases} {\rm Ai}^{(r)}_{{\cal C}_3}(\zeta),~~\zeta\in\mbox{I} \\ {\rm Ai}^{(r)}_{{\cal C}_4}(\zeta) ,~~\zeta\in \mbox{II and III} ,\\
{\rm Ai}^{(r)}_{{\cal C}_5}(\zeta),~~\zeta\in\mbox{IV}.
\\
\end{cases}
\label{48}
\end{equation}
\end{thm}

\begin{proof}
By applying the Sokhotskyi-Plemelj formula, one can satisfy the jump condition \eqref{ARHP} by defining
\begin{equation}\label{148}
\begin{split}
v^{(r)}_1(\zeta)&:={\rm Ai}^{(r)}_{{\cal C}_1}(\zeta;\tau),
\\v^{(r)}_2(\zeta)&:=\displaystyle \frac{e^{-\tau\zeta}}{2i\pi}\int_{{\mathbb R}}\frac{v_1^{(r)}(s)e^{\tau s}}{s-\zeta}ds-{\rm Ai}_{{\cal C}_4}^{(r)}(\zeta;\tau),
\\v^{(r)}_3(\zeta)&:=\displaystyle\frac{e^{-\tau\zeta}}{2i\pi}\int_{{\cal C}_6} \frac{v^{(r)}_2(t) e^{\tau t} dt}{t-\zeta}-\displaystyle\frac{e^{-\tau\zeta}}{2i\pi}\int_{{\cal C}_2} \frac{v_2^{(r)}(t) e^{\tau t} dt}{t-\zeta}.
\end{split}
\end{equation}
Here we note that ${\rm Ai}_{{\cal C}_4}^{(r)}(\zeta;\tau)$ is holomorphic (therefore having no jump).  
The next step is to verify that the definitions in \eqref{148} are equivalent to the advocated form in \eqref{47} and \eqref{48}.

First let us consider $v^{(r)}_2(\zeta)$.
\begin{align}
v^{(r)}_2(\zeta)&=\frac{e^{-\tau\zeta}}{2i\pi}\int_{\mathbb R}\frac{ds}{s-\zeta}\frac{1}{2i\pi}\left(\int_{-\tau-i\infty}^{-\tau}\!\!\!\!\!\!\!\!\!\! dt+\int_{-\tau}^{i\infty-\tau}\!\!\!\!\!\!dt\right) (t+\tau)^r e^{s(t+\tau)-t^3/3}-{\rm Ai}_{{\cal C}_4}^{(r)}(\zeta;\tau)
\\&=\frac{e^{-\tau\zeta}}{(2i\pi)^2}\left(\int_{\mbox{\bf sc}-}\!\!\!\!ds\int_{-\tau-i\infty}^{-\tau}\!\!\!\!\!\!\!\!\!\! dt+\int_{\mbox{\bf  sc}\scriptstyle{\bf +}}\!\!\!\!\!\!ds\int_{-\tau}^{i\infty-\tau}\!\!\!\!\!\!dt\right)\frac{(t+\tau)^r e^{s(t+\tau)-t^3/3}}{s-\zeta}
+\frac{1}{2i\pi}\int_\infty^{-\tau}\!\!(t+\tau)^re^{\zeta t-t^3/3}dt
\\&=\begin{cases}\displaystyle\frac{1}{2i\pi}\int_{{\cal C}_2}(t+\tau)^r e^{\zeta t-t^3/3}dt={\rm Ai}^{(r)}_{{\cal C}_2}(\zeta;\tau), ~~&\zeta\in{\mathbb H}^+,\vspace{0.1cm}
\\\displaystyle\frac{1}{2i\pi}\int_{{\cal C}_6}(t+\tau)^r e^{\zeta t-t^3/3}dt={\rm Ai}^{(r)}_{{\cal C}_6}(\zeta;\tau), &\zeta\in{\mathbb H}^-,\end{cases}\label{v2}
\end{align}
where, at the second equality, {\bf sc}- and {\bf sc}+ are respectively the completions of the contour ${\mathbb R}$ by adding the infinite half-circle through the lower ($-$) half-plane and the upper ($+$) half-plane; {\bf sc} stands for ``semi-circle."

Using this result, $v^{(r)}_3(\zeta)$ becomes
\begin{equation}
\begin{split}
v^{(r)}_3(\zeta)&=\frac{e^{-\tau\zeta}}{2i\pi}\left(\int_{{\cal C}_6} \frac{{\rm Ai}^{(r)}_{{\cal C}_6}(t;\tau)e^{\tau t} dt}{t-\zeta}-\int_{{\cal C}_2} \frac{{\rm Ai}^{(r)}_{{\cal C}_2}(t;\tau)e^{\tau t} dt}{t-\zeta}\right)
\\&=
\frac{e^{-\tau\zeta}}{2i\pi}\left(\int_{{\cal C}_6} \frac{e^{\tau t}\left({\rm Ai}^{(r)}_{{\cal C}_5}(t;\tau)-{\rm Ai}^{(r)}_{{\cal C}_4}(t;\tau)\right)}{t-\zeta}dt+\int_{{\cal C}_2} \frac{e^{\tau t}\left({\rm Ai}^{(r)}_{{\cal C}_4}(t;\tau)-{\rm Ai}^{(r)}_{{\cal C}_3}(t;\tau)\right)}{t-\zeta}dt\right)
\\&=\frac{e^{-\tau\zeta}}{2i\pi}\left(\int_{{\cal C}_1} \frac{{\rm Ai}^{(r)}_{{\cal C}_4}(t;\tau) e^{\tau t}dt}{t-\zeta}+\int_{{\cal C}_6} \frac{{\rm Ai}^{(r)}_{{\cal C}_5}(t;\tau) e^{\tau t}dt}{t-\zeta}-\int_{{\cal C}_2} \frac{{\rm Ai}^{(r)}_{{\cal C}_3}(t;\tau) e^{\tau t}dt}{t-\zeta}\right) .
\end{split}
\end{equation}
It can be noticed that, in all three terms in the last expression, one can close the contours ${\cal C}_1,{\cal C}_2$, and ${\cal C}_6$ by adding the corresponding arcs of infinite radius.  For instance, ${\cal C}_1$ can be made a closed contour by adding the arc $\{r e^{\theta}|\theta\in [2\pi/3,4\pi/3]; r\to\infty\}$.   Such addition is allowed because the term $e^{(s+\tau)t}$ in 
\eq
\frac{{\rm Ai}^{(r)}_{{\cal C}_4}(t;\tau) e^{\tau t}}{t-\zeta}=\frac{1}{2\pi i
(t-\zeta)}\int_{{\cal C}_4}(s+\tau)^re^{t(s+\tau)-s^3}ds
\endeq
is suppressed over the infinite arc when $s$ is integrated over ${\cal C}_4$.  The other terms work the same.  Then each integration on the closed contour becomes a residue calculation that leads to the definitions in the theorem. \end{proof}

\subsection{Asymptotic behavior of 
\texorpdfstring{$\bs{\mathcal{A}_r(\zeta)}$}{Ar} for large 
\texorpdfstring{$\bs{\zeta}$}{aweq} and growing 
\texorpdfstring{$\bs{r}$}{qwie}}

The matrix $\mathcal{A}_r(\zeta)$ will be used to build an appropriate 
local parametrix near $z=\beta$.  The resulting global parametrix 
will have a jump across $\partial\mathbb{D}_c$.  To control this 
jump it is necessary to see how $\mathcal{A}_r(\zeta)$ behaves 
as $\zeta\to\infty$ (recall from \eqref{zeta-N-exp} that 
$\zeta=\mathcal{O}(n^{1/6})$ on $\partial\mathbb{D}_c$).  Special 
care is needed because $r$ is also growing with $n$ as $r=\mathcal{O}(n^\gamma)$ for $0\leq\gamma<1/12$.

Let $\mathcal{A}_{(r,ij)}$ 
stand for the (i,j) entry of $\mathcal{A}_r(\zeta)$.  We will show in 
Section \ref{kernel-section} that the leading-order asymptotics of the 
kernel are smooth near $z=\beta$.  
Therefore it is sufficient to restrict ourselves to $\zeta$ in region $I$.  
The proofs of the following two propositions are given in Appendix 
\ref{proofs-appendix}.  

\begin{prop}
\label{Ar-col12}
As $\zeta\to\infty$ with $\zeta\in I$ and for $r=\mathcal{O}(n^\gamma)$, 
$0\leq\gamma<1/12$, the 
entries of the 
first and second columns of $\mathcal{A}_r(\zeta)$ behave as follows:
\begin{align}\label{arbitrary-v1}
A_{r, 11}(\zeta) &= \mbox{Ai}_{\mathcal{C}_1}^{(r)}(\zeta; \tau) = 
\frac{(-1)^r}{2 \sqrt{\pi}}\left( \zeta^{1/2} - \tau\right)^r \zeta^{-1/4} e^{-\frac{2}{3} \zeta^{3/2}} \left[ 1 + \sum_{j=3}^{6M+5} A^{(r, 11)}_j \zeta^{-j/2} + 
\mathcal{O}\left( \frac{r^{4M+4}}{\zeta^{3M+3}} \right) \right], \\
\label{arbitrary-21}
A_{r, 21}(\zeta) &=\partial_\zeta \mbox{Ai}_{\mathcal{C}_1}^{(r)}(\zeta; \tau) = \frac{(-1)^r}{2\sqrt{\pi}} (\zeta^{1/2}-\tau)^r \zeta^{ 1/4} e^{-\frac{2}{3} \zeta^{3/2}} \left[ -1 + \sum_{j=3}^{6M+5} A_j^{(r, 21)} \zeta^{-j/2} + \mathcal{O}\left( \frac{r^{4M+4}}{\zeta^{3M+3}} \right) \right], \\
\begin{split}
\label{arbitrary-31}
A_{r, 31}(\zeta) &=  \mbox{Ai}_{\mathcal{C}_1}^{(r-1)}(\zeta; \tau) = 
\frac{(-1)^{r-1}}{2\sqrt{\pi}} (\zeta^{1/2} - \tau)^r e^{-\frac{2}{3} \zeta^{3/2}} \left[ \frac{1}{\zeta^{3/4}} + \frac{\tau}{\zeta^{5/4}} + \frac{\tau^2}{\zeta^{7/4}} + \sum_{j=3}^{6M+5} A_j^{(r, 31)} \zeta^{-3/4 - j/2} \right. \\ 
&\phantom{=  \mbox{Ai}_{\mathcal{C}_1}^{(r-1)}(\zeta; \tau) = 
\frac{(-1)^{r-1}}{2\sqrt{\pi}} (\zeta^{1/2} - \tau)^r e^{-\frac{2}{3} \zeta^{3/2}} }
\left. + \mathcal{O}\left( \frac{r^{4M+4}}{\zeta^{3M+15/4}} \right) \right],
\end{split}\\
 \label{arbitrary-v2}
\mathcal{A}_{(r,12)}(\zeta) & = \mbox{Ai}_{\mathcal{C}_2}^{(r)}(\zeta; \tau) = -\frac{1}{2\sqrt{\pi} i} \left(\zeta^{1/2} + \tau\right)^r \zeta^{-1/4} e^{\frac{2}{3}\zeta^{3/2} } \left[ 1 + \sum_{j=3}^{6M+5} A_j^{(r, 12)} \zeta^{-j/2} + 
\mathcal{O}\left( \frac{r^{4M+4}}{\zeta^{3M+3}} \right) \right], \\
\label{arbitrary-22}
\mathcal{A}_{(r,22)}(\zeta) & = \partial_\zeta \mbox{Ai}_{\mathcal{C}_2}^{(r)}(\zeta;\tau) = -\frac{1}{2\sqrt{\pi} i} \left(\zeta^{1/2}+\tau\right)^r \zeta^{1/4} e^{\frac{2}{3} \zeta^{3/2} } \left[ 1 + \sum_{j=3}^{6M+5} A_j^{(r, 22)} + 
\mathcal{O}\left( \frac{r^{4M+4}}{\zeta^{3M+3}} \right) \right], \\
\begin{split}
\label{arbitrary-32}
\mathcal{A}_{(r,32)}(\zeta) & = \mbox{Ai}_{\mathcal{C}_2}^{(r-1)}(\zeta; \tau) = - \frac{1}{2\sqrt{\pi} i} \left(\zeta^{1/2}+\tau\right)^{r} e^{\frac{2}{3}\zeta^{3/2} }
\left[ \frac{1}{\zeta^{3/4}} - \frac{\tau}{\zeta^{5/4}} + \frac{\tau^2}{\zeta^{7/4}} + \sum_{j=3}^{6M+5} A_j^{(r, 32)} \zeta^{-3/4-j/2} \right. \\
& \phantom{ = \mbox{Ai}_{\mathcal{C}_2}^{(r-1)}(\zeta; \tau) = - \frac{1}{2\sqrt{\pi} i} \left(\zeta^{1/2}+\tau\right)^{r} e^{\frac{2}{3}\zeta^{3/2} }}
\left. +
 \mathcal{O}\left( \frac{r^{4M+4}}{\zeta^{3M+15/4}} \right) \right], 
\end{split}
\end{align}
where the $A_j^{(r, ik)}$ are polynomial functions of $r$ and 
$\tau$.  
\end{prop}
We will not need the exact form of the coefficients $A_j^{(r,k\ell)}$, but they can be determined algorithmically
through the proof in Appendix \ref{proofs-appendix}.

\begin{prop}
\label{Ar-col3}
The 
entries of the third column of $\mathcal{A}_r(\zeta)$ behave as follows as $\zeta\to\infty$ and for $r=\mathcal{O}(n^\gamma)$, $0\leq\gamma<1/12$:
\begin{equation}
\label{Ar13}
\mathcal{A}_{(r,13)}(\zeta)=\frac{(-1)^{r+1}r!}{2 i\pi}\frac{{\rm e}^{-\tau \zeta}}{\zeta^{r+1}}\left(\sum_{j=0}^{M} \frac{\eta_{j}(\tau)}{\zeta^{j}}+  {\cal O}\left( \frac{r^{M+1}}{|\zeta|^{M+1}}\right) \right),\qquad \zeta\in \partial {\mathbb D}_{c},
\end{equation}
where 
\begin{equation}
\label{eta-of-tau}
\eta_{j}(\tau):=\sum_{(A,B)}\int_{{\cal C}_B}{\rm Ai}_{{\cal C}_A}^{(0)}(t;\tau){\rm e}^{\tau t} t^{j} d t.
\end{equation}
The asymptotic expansion of ${\cal A}_{(r,23)}=\partial v_{3}^{(r)}$ is given in the same form but with different coefficients:
\begin{equation}
\eta_{j}(\tau)\to \widehat\eta_{j}(\tau):= -\tau \eta_{j}(\tau)-(r+1+j)\eta_{j-1}(\tau),\qquad \eta_{-1}(\tau):=0.
\end{equation}
The asymptotic expansion of ${\cal A}_{r,33}$ is given by ${\cal A}_{(r-1,13)}$ using \eqref{Ar13}.
\end{prop}
As a side remark, by applying the rotational symmetry of the Airy function, we get (\ref{eta-naught}) repeated here:
\begin{equation}
\eta_0(\tau)=\int_{0}^\infty {\rm Ai}(t)\left(e^{\tau t/\omega}+e^{\omega\tau t}+e^{\tau t}\right) dt,\quad \omega:=\exp\left(\frac{2\pi i}{3}\right).
\end{equation}
Evaluating the integral, $\eta_0(\tau)$ is $1$ at $\tau=0$.

The first few terms in the leading-order expansion of the entries of 
$\mathcal{A}_{(r)}(\zeta)$ are given in Section \ref{Ar-explicit}.

For $r=
{{\cal O}(n^{\gamma}), 0\leq\gamma<1/12}$ and $\zeta\in I$, Propositions \ref{Ar-col12} and 
\ref{Ar-col3} yield 
immediately the following expression for $\mathcal{A}_r(\zeta)$:
\eq
\label{L0-of-zeta}
\mathcal{A}_r(\zeta) = \bpm n^{-1/6} & 0 & 0 \\ 0 & n^{1/6} & 0 \\ 0 & 0 & 1 \epm {\bf L}_{(0)}(\zeta) {\bf S}(\zeta) {\bf C}(\zeta),
\endeq
where ${\bf S}(\zeta)$ is given by \eqref{S-of-zeta}, 
\eq
{\bf C}(\zeta) := \bpm \frac{1}{\sqrt{\pi}} \zeta^{r/2} e^{-\frac{2}{3}\zeta^{3/2}} & 0 & 0 \\ 0 & \frac{1}{\sqrt{\pi}} \zeta^{r/2} e^{\frac{2}{3} \zeta^{3/2}} & 0 \\ 0 & 0 & \frac{1}{2\pi i}\Gamma(r)\eta(\tau)\zeta^{-r}e^{-\tau \zeta} \epm, 
\endeq
and 
\eq
\label{L0-zeta-exp}
{\bf L}_{(0)}(\zeta) = {\bf I} + \bpm  \mathcal{O}\left(\frac{r^2}{\zeta}\right)  &  \mathcal{O}\left(\frac{r n^{1/3}}{\zeta}\right)  &  \frac{-rn^{1/6}}{\zeta} + \mathcal{O}\left(\frac{r^2 n^{1/6}}{\zeta^2}\right)  \\  \mathcal{O}\left(\frac{r}{n^{1/3}}\right) & \mathcal{O}\left(\frac{r^2}{\zeta}\right) & \mathcal{O}\left(\frac{r}{\zeta n^{1/6}}\right) \\ \mathcal{O}\left(\frac{r}{\zeta n^{1/6}} \right) & \mathcal{O}\left(\frac{n^{1/6}}{\zeta}\right) & \mathcal{O}\left(\frac{r}{\zeta}\right)  \epm.
\endeq
There are several important points to note concerning ${\bf L}_{(0)}(\zeta)$, 
which we will use to define the global parametrix inside $\mathbb{D}_c$.  
We are specifically interested in the behavior of ${\bf L}_{(0)}(\zeta)$ for 
$\zeta\in\partial\mathbb{D}_c$, where $\zeta=\mathcal{O}(n^{1/6})$.  
The observant reader will note that several of the terms in 
${\bf L}_{(0)}(\zeta)$ are not decaying as $n\to\infty$ for 
$\zeta\in\partial\mathbb{D}_c$.  Ideally we would have 
${\bf L}_{(0)}(\zeta)$ close to the identity for large $n$.  This will 
be achieved by a series of \emph{pseudo-Schlesinger transforms} (from now on, we simply use the term ``Schlesinger transform").  Each 
transform only impacts one column at a time and 
will modify ${\bf L}_{(0)}$ into a  new ${\bf L}_{(j)}$ where the most 
dominant term in that column of the expansion \eqref{L0-zeta-exp} is 
sequentially removed.  At each step the rate of convergence for that 
column will be improved.

We will show in detail how 
such a transform is used to define a new matrix ${\bf L}_{(1)}(\zeta)$ 
with the same entry-wise growth as ${\bf L}_{(0)}(\zeta)$ except that the 
highest-order term in the third column (the explicit $-rn^{1/6}/\zeta$ term 
in the $(1,3)$ entry) is removed. 

A note is in order concerning the $\zeta$-independent term in the $(2,1)$ 
entry.  There is a nilpotent matrix ${\bf N}_0$ such that 
${\bf L}_{(0)}{\bf N}_0$ is the same as ${\bf L}_{(0)}$ except this term 
is removed.  
We then redefine ${\bf H}_{(0)}:={\bf H}_{(0)}{\bf N}_0$ and ${\bf L}_{(0)}:={\bf L}_{(0)}{\bf N}_0$ as our starting point.

We then apply a finite sequence of Schlesinger transforms to the 
improved ${\bf L}_{(1)}(\zeta)$ matrix to yield a final matrix 
${\bf L}_{(2)}(\zeta)$ 
of the form ${\bf I}+\mathcal{O}(n^{-2/3})$.  In fact it is possible to 
use such procedure to change ${\bf L}_{(0)}(\zeta)$ into a 
matrix arbitrarily close to the identity.  However, other factors in 
the calculation (see \eqref{H2invH2}) will introduce an 
$\mathcal{O}(n^{-2/3})$ error into the 
kernel computation, so there is no point in removing smaller terms.  

We start with two observations.  First off, ${\bf L}_{(0)}(\zeta)$ has no jump discontinuities, and the expansion of 
   ${\bf L}_{(0)}(\zeta)$ has only negative integer powers of 
$\zeta$.  This is important since the Schlesinger transform will 
remove pole terms.  Secondly, ${\bf L}_{(0)}(\zeta)$ has an asymptotic 
expansion for large $n$.  More exactly, once the pole of the lowest order is 
removed from a given entry, the remaining terms are asymptotically smaller 
as $n\to\infty$ than what was removed.  Furthermore, each entry can be 
made to decay at an arbitrarily fast rate by removing a finite number 
of terms from the Laurent expansion of that entry.  This follows from the 
behavior of the entries of $\mathcal{A}_r(\zeta)$;  see Propositions 
\ref{Ar-col12} and \ref{Ar-col3}.  
The leading-order entries of ${\bf L}_{(0)}(\zeta)$ can be read off from 
the formulas in Section \ref{Ar-explicit}, and in principle any term can be 
computed as described in the proofs.    

Now is an appropriate time to explain where the limitation $\gamma<1/12$ 
arises (recall $r=\mathcal{O}(n^\gamma)$ for $0\leq\gamma<1/12$).  Note in Proposition 
\ref{Ar-col12} the terms $(\zeta^{1/2}\pm\tau)^r$.  To explicitly compute 
the entries of ${\bf L}_{(0)}(\zeta)$, it is necessary to use the expansions 
\begin{equation}
\left(1 \pm \frac{\tau}{\zeta^{1/2}} \right)^{r}  = 1 \pm \frac{r\tau}{\zeta^{1/2}} + \frac{r(r-1) \tau^2}{2 \zeta} \pm \frac{r (r-1) (r-2) \tau^3}{3! \zeta^{3/2} }+ \mathcal{O}\left( \frac{r^4}{\zeta^2} \right),
\end{equation}
which are asymptotic expansions for $\zeta\in\partial\mathbb{D}_c$ only 
for $\gamma<1/12$.  Note that the number of terms in ${\bf L}_{(0)}(\zeta)$ 
that need to be removed to leave an error of $\mathcal{O}(n^{-2/3})$ will 
depend on how fast $r$ is growing.  Indeed, if $\gamma$ is close to $1/12$, 
then a very large number of terms will need to be removed!

\section{The error analysis \label{Error-Section}}

\subsection{Global Parametrix Construction}

The asymptotic expansion of the outer parametrix ${\bf \Psi}(z)$
on the boundary of $\mathbb{D}_c$ (Lemma \ref{H-lemma}) will guide our first
attempt at defining the \emph{global parametrix} ${\bf \Psi}_{(0)}^\infty$.  Around all of the 
band endpoints except for $\beta$ we will use standard Airy parametrices 
which match with ${\bf \Psi}(z)$ up to $\mathcal{O}(1/n)$.  For more 
details on the Airy parametrix in a $3\times 3$ Riemann-Hilbert problem see, 
for example, Bleher and Kuijlaars \cite{Bleher:2004b}.  For convenience, 
we will denote the Airy parametrices by ${\bf P}_{\mbox{Ai}}(z)$.  
We also define $\mathbb{D}_{\mbox{Ai}}$ to be the union of $2g+1$ small, 
fixed-size disks centered around $\alpha_1, \beta_1, ..., \alpha_g, \beta_g$, 
and $\alpha$ in which the Airy parametrices will be used.  We set
\eq
\label{Psi0-infinity}
{\bf \Psi}_{(0)}^\infty(z) := \begin{cases} {\bf \Psi}(z) = {\bf H}_{(0)}(z){\bf S}(\zeta) & z\notin(\mathbb{D}_c\cup\mathbb{D}_{\mbox{Ai}}), \\ {\bf H}_{(0)}(z){\bf D}(z)^{-1} {\bf L}_{(0)}(\zeta){\bf S}(\zeta){\bf D}(z) & z\in\mathbb{D}_c, \\ {\bf P}_{\mbox{Ai}}(z) & z\in\mathbb{D}_{\mbox{Ai}}, \end{cases}
\endeq
where ${\bf H}_{(0)}(z)$ is defined by Lemma \ref{H-lemma}, 
${\bf L}_{(0)}(\zeta)$ is defined by \eqref{L0-of-zeta}, ${\bf S}(\zeta)$ is 
defined by \eqref{S-of-zeta}, and 
\eq
\label{D-of-z}
{\bf D}(z):=\bpm 1 & 0 & 0 \\ 0 & 1 & 0 \\ 0 & 0 & e^{nh(z)-nh(\beta)-\tau\zeta} \epm.
\endeq

Direct calculation shows that ${\bf \Psi}_{(0)}^\infty(z)$
satisfies exactly the same jumps as ${\bf W}(z)$ inside $\mathbb{D}_c$:
\begin{equation}\label{Psiinfjumplocal}
\left\{\begin{array}{ll}\vspace{0.2cm}
{\bf \Psi}_{(0)+}^\infty(z)={\bf \Psi}_{(0)-}^\infty(z)\bpm 1 & e^{-\frac43\zeta^{3/2}} & e^{-\frac23\zeta^{3/2}+nh(z)}\\ 0 & 1 & 0 \\ 0 & 0 & 1 \epm & z\in \Gamma_1\cap\mathbb{D}_c,
\\\vspace{0.2cm}
{\bf \Psi}_{(0)+}^\infty(z)={\bf \Psi}_{(0)-}^\infty(z)\bpm 1&0&0\\-e^{\frac43\zeta^{3/2}}&1&-e^{\frac23\zeta^{3/2}+nh(z)}\\0&0&1\epm & z\in \Gamma_2\cap\mathbb{D}_c,
\\\vspace{0.2cm}
{\bf \Psi}_{(0)+}^\infty(z)={\bf \Psi}_{(0)-}^\infty(z)\bpm 0&
{(-1)^r}
&0\\
{(-1)^{r+1}}
&0&0\\0&0&1\epm & z\in\Gamma_3\cap\mathbb{D}_c,
\\\vspace{0.2cm}
{\bf \Psi}_{(0)+}^\infty(z)={\bf \Psi}_{(0)-}^\infty(z)\bpm 1&0&0\\e^{\frac43\zeta^{3/2}}&1&-e^{\frac23\zeta^{3/2}+nh(z)}\\0&0&1\epm & z\in\Gamma_4\cap\mathbb{D}_c.
\end{array}\right.
\end{equation}
We now define the \emph{error matrix} ${\bf E}_{(0)}(z)$ by
\eq
{\bf E}_{(0)}(z) := {\bf W}(z)\left({\bf \Psi}_{(0)}^\infty(z)\right)^{-1}.
\endeq
Note that ${\bf E}_{(0)}(z)$ has no jumps on the contours inside
$\mathbb{D}_c$.  We take the boundary of $\mathbb{D}_c$ to be oriented
clockwise (i.e. inside/outside is $-/+$ respectively).  Then, using the expansion \eqref{Psi-ito-H}, the jump for
${\bf E}_{(0)}(z)$ is given on $\mathbb{D}_c$ by
\eq
{\bf V}_{(0)}^{({\bf E})}(z) 
:={\bf E}_{(0)}(z)^{-1}_-{\bf E}_{(0)}(z)_+={\bf\Psi}^\infty_{(0)}(z)_-{\bf\Psi}^\infty_{(0)}(z)_+^{-1}= {\bf H}_{(0)}(z) {\bf D}(z)^{-1} {\bf L}_{(0)}(\zeta) {\bf D}(z) {\bf H}_{(0)}(z)^{-1}.
\endeq
We would like this jump matrix to be close to the identity as $n\to\infty$ 
for $\zeta\in\mathbb{D}_c$.  However, 
as noted previously, ${\bf L}_{(0)}(\zeta)$ does not decay to the identity 
for $\zeta=\mathcal{O}(n^{1/6})$ (see \eqref{L0-zeta-exp}).  To remedy this
situation, we will modify ${\bf L}_{(0)}(\zeta)$ in a series of 
Schlesinger transforms:
\eq
{\bf L}_{(0)}(\zeta) \to {\bf L}_{(1)}(\zeta) \to {\bf L}_{(2)}(\zeta). 
\endeq
The first step will be carried out explicitly to explain the procedure and 
will remove one representative term.  The second step will remove the 
remaining error terms up to $\mathcal{O}(n^{-2/3})$.  There is no need to 
carry out these transforms explicitly as we will show they only affect the 
subdominant terms in the kernel.  We will then use ${\bf L}_{(2)}(\zeta)$ 
to define a refined global parametrix ${\bf \Psi}_{(2)}^\infty(z)$.  
The new error matrix ${\bf E}_{(2)}(z)$ defined using 
${\bf \Psi}_{(2)}^\infty(z)$ will be shown to be close to ${\bf I}$.  

Our first transform  removes the term $-rn^{1/6}/\zeta$ in the (13) entry 
of ${\bf L}_{(0)}$:  define a new matrix 
\eq
{\bf L}_{(1)}(\zeta):={\bf L}_{(0)}(\zeta) \bpm 1 & 0 & \frac{rn^{1/6}}{\zeta} \\ 0 & 1 & 0 \\ 0 & 0 & 1 \epm, \quad \zeta\in\mathbb{D}_c.
\endeq
Then it follows that 
\eq
\label{L1}
{\bf L}_{(1)}(\zeta) = \bpm L_{(0,11)} & L_{(0,12)} & \mathcal{O}\left(\frac{r^3 n^{1/6}}{\zeta^2} \right) \\ L_{(0,21)} & L_{(0,22)} & L_{(0,23)} + \mathcal{O}\left(\frac{r^2}{n^{1/6}\zeta}\right) \\ L_{(0,31)} & L_{(0,32)} & L_{(0,33)} + \mathcal{O}\left(\frac{r^2}{\zeta^2}\right) \epm.
\endeq
Now there is a unique $3\times 3$ matrix ${\bf F}_{(1)}$, independent of 
$z$, such that 
\eq
\begin{split}
{\bf H}_{(1)}(z)  :=  &  \left({\bf I} + \frac{{\bf F}_{(1)}}{z-\beta}\right){\bf H}_{(0)}(z){\bf D}(z)^{-1} \bpm 1 & 0 & \frac{rn^{1/6}}{\zeta} \\ 0 & 1 & 0 \\ 0 & 0 & 1 \epm {\bf D}(z) \\
  = & \left({\bf I} + \frac{{\bf F}_{(1)}}{z-\beta}\right){\bf H}_{(0)}(z) \bpm 1 & 0 & e^{nh-\tau\zeta}\frac{rn^{1/6}}{\zeta} \\ 0 & 1 & 0 \\ 0 & 0 & 1 \epm
\end{split}
\endeq
is holomorphic for $z\in{\mathbb D}_c$.  The matrix ${\bf F}_{(1)}$ 
can be computed explicitly, but since its exact form will not be 
significant we use the following proposition to assert its existence.  

\begin{prop}
\label{theoremSch}
Fix a function $\epsilon(n)$ so $\lim_{n\to\infty}\epsilon(n)=0$.  Assume 
that we are given a constant $q\times q$ two-nilpotent matrix ${\bf N}$
(i.e. ${\bf N} ^2=0$)
, a 
series 
of numbers $\{d_j(n)\big|-\infty<j\leq k\}$ so that 
$d_k\neq0$ and 
$d_j=\mathcal{O}(\epsilon)$, and a $q\times q$ matrix ${\bf H}(z;n)$ 
that is locally holomorphic at $z=0$ and $\det{\bf H}(z;n)\equiv 1$.  
We also assume that 
${\bf H}(z;n)-\lim_{n\to\infty}{\bf H}(z;n)={\cal O}(\epsilon)$
uniformly on a fixed, finite disk around $z=0$.
Then one can uniquely determine $k$ constant (in $z$) matrices ${\bf F}_1,...,{\bf F}_k$ by requiring that $\widetilde{\bf H}(z;n)$ defined below is locally holomorphic at $z=0$:
\begin{equation}\label{tildeH}
\widetilde{\bf H}(z;n):=\left({\bf I}+\frac{{\bf F}_1}{z}+...+\frac{{\bf F}_k}{z^k}\right){\bf H}(z;n)\left({\bf I}+ {\bf N}\sum_{j=-\infty}^k\frac {d_j}{z^j}\right).
\end{equation}
In addition, 
\eq
\label{Fj-bound}
{\bf F}_j(n) = \mathcal{O}(\epsilon), \quad j=1,...,k
\endeq
and
\eq
\widetilde{\bf H}(z;n)-\lim_{n\to\infty}\widetilde{\bf H}(z;n)={\cal O}(\epsilon)
\endeq
uniformly in the same disk around $z=0$, and $\det \widetilde{\bf H}(z;n)\equiv 1$.
\end{prop}
\begin{proof}
Denote as follows the expansion at $z=0$:
\begin{equation}
{\bf H}(z;n)={\bf H}_0(n)+{\bf H}_1(n)z+{\bf H}_2(n)z^2+\cdots
\end{equation}
Let us collect all the terms of the order $z^{-m}$ from \eqref{tildeH}.
\begin{equation}\label{three}
\mbox{Terms of order }z^{-m}: ~\sum_{l=0}^{k-m} {\bf F}_{l+m} {\bf H}_{l}+ \sum_{j=1-k}^{k}\left(\sum_{l=0}^{k-(m-j)} {\bf F}_{l+m-j} {\bf H}_{l}\right){\bf N}d_j+ \sum_{l=0}^{k-m} {\bf H}_l{\bf N} d_{l+m}.
\end{equation}
This must be zero for $m=1,...,2k$ for $\widetilde{\bf H}(z;n)$ to be holomorphic at the origin.  Since we have only $k$ unknown matrices: ${\bf F}_1,...,{\bf F}_k$, the number of equations must be reduced to $k$ equations.   We will show that the equations for $m=k+1,...,2k$ are contained in the equations for $m=1,...,k$.  
\\
For $m>k$ the first and the last (summation) terms are absent and we have only the middle term (with double summations).   The middle term is a linear combination of $\left\{\left(\sum_{l=0}^{k-\widetilde m} {\bf F}_{l+\widetilde m} {\bf H}_{l}\right){\bf N}\big|\widetilde m=1,...,k\right\}$ (which gives an invertible linear system of equations because $d_k\neq 0$) and, therefore, the set of equations are equivalent to 
\begin{equation}
0=\left(\sum_{l=0}^{k-m} {\bf F}_{l+m} {\bf H}_{l}\right){\bf N}\quad\mbox{for }m=1,...,k.
\end{equation}
These are also obtained by right multiplication of \eqref{three} for $m=1,...,k$ by ${\bf N}$ because the second and the last terms vanish given that ${\bf N}{\bf N}=0$.  Therefore the vanishing of \eqref{three} for $m=1,...,k$ is a sufficient condition to solve for ${\bf F}_j$'s. 
To solve these $k$ equations consider a big (block) matrix of size $qk \times qk$ (where $q$ is the size of the matrices that appear in \eqref{three}; $q=3$  in  our case) obtained
  by adjoining the
 ${\bf F}_j$'s side by side as $[{\bf F}_1,{\bf F}_2,...,{\bf F}_k]$.  Then we can write the $k$ equations into a single matrix equation as follows.
\begin{equation}
\begin{split}
[{\bf F}_1,{\bf F}_2,...,{\bf F}_k]&\left(\left[\begin{array}{cccc}{\bf H}_0 & 0 & \cdots & 0 \\{\bf H}_1 & {\bf H}_0 &  & 0 \\\vdots &  & \ddots & \vdots \\{\bf H}_{k-1} & {\bf H}_{k-2} & \cdots & {\bf H}_0\end{array}\right]+{\cal O}(\epsilon)\right)
\\&=-\left[ \sum_{l=0}^{k-1} {\bf H}_l{\bf N} d_{l+1},\sum_{l=0}^{k-2} {\bf H}_l{\bf N} d_{l+2},...,{\bf H}_0{\bf N}d_{k-1}+{\bf H}_1{\bf N}d_k,{\bf H}_0{\bf N} d_{k}\right],
\end{split}
\end{equation}
where the middle term of \eqref{three} is hidden in ${\cal O}(\epsilon)$.  Because $\det{\bf H}(z;n)=\det {\bf H}_0=1$, the big matrix of size $qk\times qk$ multiplied on the right of $[{\bf F}_1,...,{\bf F}_k]$ is invertible and, therefore, the solution can be uniquely obtained. Also it follows immediately that $[{\bf F}_1,...,{\bf F}_k]={\cal O}(\epsilon)$.  From this, it also follows that $\widetilde {\bf H}(z;n)-\lim_{n\to\infty}\widetilde {\bf H}(z;0)$ is uniformly bounded by ${\cal O}(\epsilon)$.  

%

Lastly, to show that $\det\widetilde{\bf H}(z;n)\equiv1$, we take the determinant of \eqref{tildeH}.  
\begin{align}\label{tildeH1}
\det\widetilde{\bf H}(z;n)&=\det\left[\left({\bf I}+\frac{{\bf F}_1}{z}+...+\frac{{\bf F}_k}{z^k}\right)\right]\det{\bf H}(z;n)\det\Big[\Big({\bf I}+{\bf N}\sum_{j=-\infty}^k\frac {d_j}{z^j}\Big)\Big]
\\&=\det\left[\left({\bf I}+\frac{{\bf F}_1}{z}+...+\frac{{\bf F}_k}{z^k}\right)\right].
\end{align}
The left hand side is holomorphic (at the origin) and therefore it must be 1 to match the right hand side.
\end{proof}
We now show how to apply Proposition \ref{theoremSch} to guarantee the existence 
of the matrix ${\bf F}_{(1)}$.  First, using the holomorphicity of 
$\zeta(z)$ in 
$\mathbb{D}_c$, we see there are unique 
functions $d_{(1,j)}(n)$ such that
\eq
e^{nh-\tau\zeta}\frac{rn^{1/6}}{\zeta} = \sum_{j=-\infty}^1 \frac {d_{(1,j)}(n)}{(z-\beta)^j}.
\endeq
Furthermore, using \eqref{zeta-ito-zmbeta} and \eqref{hz1},
\eq
d_{(1,j)}(n) = \mathcal{O}\left(\frac{r}{n^{1/2}}\right).
\endeq
We shift $z$ by $\beta$ and choose 
\eq
{\bf H}(z;n) := {\bf H}_{(0)}(z;n), \quad \epsilon(n):=\frac{r}{n^{1/2}}, \quad {\bf N}:=\bpm 0 & 0 & 1 \\ 0 & 0 & 0 \\ 0 & 0 & 0 \epm.
\endeq
Recall from Lemma \ref{H-lemma} that the determinant of ${\bf H}_{(0)}$ is 
one and that 
${\bf H}_{(0)}(z;n)-\lim_{n\to\infty}{\bf H}_{(0)}(z;n)=\mathcal{O}(\kappa)$.  Therefore, Proposition  \ref{theoremSch} shows that  the matrices 
${\bf F}_{(1,m)}$ exist and 
\eq\label{H1bound}
{\bf H}_{(1)}(z;n) - \lim_{n\to\infty}{\bf H}_{(1)}(z;n) = \mathcal{O}\left(\frac{r}{n^{1/2}}\right).
\endeq 

%

Furthermore, any series expansion of a unimodular matrix can be decomposed 
into products of the form $[{\bf I}+{\bf N}]$, where ${\bf N}$ is nilpotent as we show below in Lemma \ref{nilpo}
These products can then be dealt with as explained in 
Proposition 
\ref{theoremSch}.
\begin{lemma}
\label{nilpo}
  Consider a unimodular matrix ${\bf M}(z)$ with expansion
\begin{equation}
{\bf M}(z)={\bf I}+{\bf M}_{1}/z+{\bf M}_{2}/z^{2}+{\bf M}_{3}/z^{3}+...,\qquad 
{\bf M}_j\in {\rm Mat}_{3\times 3}(\C),\ \ \text{\rm Tr}\, {\bf M}_1=0.
\end{equation} 
Then, for arbitrary $K>0$ we can find a finite number of nilpotent matrices $\mathbf N_j$ and integers $k_j\geq 1$ such that 
\be
\mathbf M(z) = \prod_{j} \le(\mathbf I +\frac{ \mathbf N_j}{z^{k_j}} \ri) \le(\mathbf I + \mathcal O(z^{-K-1})\ri)\ .
\ee
\end{lemma}

\begin{proof}
The matrix ${\bf M}_{1}$ can be expressed as a linear combination of the following 
nilpotent matrices:
\eq
\begin{split}
\left(\begin{array}{ccc}0 & 1 & 0 \\0 & 0 & 0 \\0 & 0 & 0\end{array}\right),\quad
\left(\begin{array}{ccc}0 & 0 & 1 \\0 & 0 & 0 \\0 & 0 & 0\end{array}\right)&,\quad 
\left(\begin{array}{ccc}0 & 0 & 0 \\0 & 0 & 1 \\0 & 0 & 0\end{array}\right),\quad
\left(\begin{array}{ccc}0 & 0 & 0 \\0 & 0 & 0 \\1 & 0 & 0\end{array}\right),
\\
\left(\begin{array}{ccc}0 & 0 & 0 \\0 & 0 & 0 \\0 & 1 & 0\end{array}\right),\quad
\left(\begin{array}{ccc}0 & 0 & 0 \\1 & 0 & 0 \\0 & 0 & 0\end{array}\right),\quad
& \left(\begin{array}{ccc}1 & 1 & 0 \\-1 & -1 & 0 \\0 & 0 & 0\end{array}\right),\quad
\left(\begin{array}{ccc}0 & 0 & 0 \\0 & 1 & 1 \\0 & -1 & -1\end{array}\right).
\end{split}
\endeq
We label these basis elements as ${\bf N}_{j}$ with $j=1,\dots,8$
so that 
\begin{equation}
{\bf M}_{1}=\sum_{j=1}^{8} c_{j}{\bf N}_{j}
\end{equation}
for some constants $\{c_{j}\}$.  Now we may decompose ${\bf M}(z)$ as
\begin{equation}
{\bf M}(z)=\left(\prod_{j=1}^{8}\left({\bf I}+\frac{c_{j}{\bf N}_{j}}{z}\right)\right)\widetilde{\bf M}(z).
\end{equation}
The factor
 $\widetilde{\bf M}(z)$ does not have a $z^{-1}$ term 
in its expansion:
\begin{equation}
\widetilde{\bf M}(z)={\bf I}+\frac{\widetilde{\bf M}_{2}}{z^{2}}+\cdots
\end{equation}
The matrix $\widetilde{\bf M}(z)$ is still unimodular and hence $\widetilde {\bf M}_2$ is traceless and can be decomposed similarly as before
\begin{equation}
\widetilde{\bf M}(z)=\left(\prod_{j=1}^{8}\left({\bf I}+\frac{\widetilde c_{j}{\bf N}_{j}}{z^{2}}\right)\right)\,\widetilde{\widetilde{\bf M}}(z).
\end{equation}
One can iterate this procedure to obtain the decomposition to any arbitrary 
order. 
\end{proof}

 We can thus apply Proposition \ref{theoremSch} to each entry in 
${\bf L}_{(1)}(\zeta)$ 
in \eqref{L1} to remove all error terms up to 
$\mathcal{O}(1/n^{2/3})$.  Let ${\bf T}(\zeta)$ be the 
appropriate transform from ${\bf L}_{(1)}(\zeta)$ to ${\bf L}_{(2)}(\zeta)$:  
\eq
{\bf L}_{(2)}(\zeta):={\bf L}_{(1)}(\zeta) {\bf T}(\zeta), \quad \zeta\in\mathbb{D}_c.
\endeq
The matrix ${\bf T}(\zeta)$ is chosen so that for $z$ on the boundary of 
$\mathbb{D}_c$ 
(where $\zeta=\mathcal{O}(n^{1/6})$), 
\eq
\label{L2-bound}
{\bf L}_{(2)}(\zeta(z)) = {\bf I} + \mathcal{O}\left(\frac{1}{n^{2/3}} \right), \quad z\in\partial\mathbb{D}_c.
\endeq
Let $p$ be the order of the highest-order 
pole in $\zeta$ which will need to be removed.
Then Proposition \ref{theoremSch} shows there are unique 
$3\times 3$ matrices ${\bf F}_{(2,1)},...,{\bf F}_{(2,p)}$, independent 
of $z$, such that
\eq
{\bf H}_{(2)}(z) := \left({\bf I}+\sum_{m=1}^p\frac{{\bf F}_{(2,m)}}{(z-\beta)^m} \right) {\bf H}_{(1)}(z) {\bf D}(z)^{-1} {\bf T}(\zeta) {\bf D}(z)
\endeq
is holomorphic for $z\in{\mathbb D}_c$.  
In addition,
\eq
\label{H2-n-decay}
{\bf H}_{(2)}(z;n) - \lim_{n\to\infty}{\bf H}_{(2)}(z;n) = \mathcal{O}\left(\frac{r}{n^{1/3}}\right)
\endeq
and
\eq\label{F2m}
{\bf F}_{(2,m)} = \mathcal{O}\left(\frac{r}{n^{1/3}}\right),
\endeq
where this $\mathcal{O}(r/n^{1/3})$ error comes from removing terms in 
the (12) entry of ${\bf L}_{(1)}(\zeta)$.
Now we can define
\eq \label{crit-psi-2}
{\bf \Psi}_{(2)}^\infty(z) := \begin{cases} {\bf R}(z) {\bf H}_{(0)}(z){\bf S}(\zeta) & z\notin(\mathbb{D}_c\cup\mathbb{D}_{\mbox{Ai}}), \\ {\bf H}_{(2)}(z){\bf D}(z)^{-1}{\bf L}_{(0)}(\zeta){\bf S}(\zeta){\bf D}(z) & z\in\mathbb{D}_c, \\ {\bf R}(z) {\bf P}_{\mbox{Ai}}(z) & z\in\mathbb{D}_{\mbox{Ai}}, \end{cases}
\endeq
where 
\begin{equation}
\label{R-def}
{\bf R}(z) := \left({\bf I} + \sum_{m=1}^p \frac{ {\bf F}_{(2,m)}}{(z-\beta)^m} \right)\left({\bf I} + \frac{ {\bf F}_{(1)}}{z-\beta}\right).
\end{equation}

The new parametrix ${\bf \Psi}_{(2)}^\infty(z)$ has the same jumps as
${\bf W}(z)$ inside ${\mathbb D}_c$.  Define the new error matrix by
\eq
\label{E2}
{\bf E}_{(2)}(z) := {\bf W}(z)\left({\bf \Psi}_{(2)}^\infty(z)\right)^{-1}
\endeq
and note that its jump for $z\in\partial\mathbb{D}_c$ is
\eq \label{v2-jump-dc}
\begin{split}
{\bf V}_{(2)}^{({\bf E})}(z) & = {\bf H}_{(2)}(z){\bf D}(z)^{-1}{\bf L}_{(0)}(\zeta){\bf D}(z){\bf H}_{(0)}(z)^{-1}{\bf R}(z)^{-1} \\
   & = {\bf H}_{(2)}(z){\bf D}(z)^{-1}{\bf L}_{(2)}(\zeta){\bf D}(z){\bf H}_{(2)}(z)^{-1}.
\end{split}
\endeq

\subsection{Global Error Computation}

We have defined the global parametrix after two Schlesinger transforms as ${\bf \Psi}_{(2)}^{\infty}(z)$ in (\ref{crit-psi-2}) and the resulting error matrix ${\bf E}_{(2)}(z)$ by \eqref{E2}.   
This matrix satisfies a Riemann-Hilbert problem with jumps across the contours pictured in Figure \ref{fig_crit} and the boundaries of $\mathbb{D}_{\mbox{Ai}}$ and $\mathbb{D}_c$.

Explicitly the jumps on these contours are:
\begin{itemize}

\item For $z$ on the contours outside of the disks $\mathbb{D}_{\mbox{Ai}}$ and $\mathbb{D}_c$, the jump of ${\bf E}_{(2)}$ is
\begin{equation}  \label{crit-version-63}
{\bf V}_{(2)}^{\bf (E)}(z) = {\bf \Psi}_{(2)}^\infty(z) {\bf V}^{\bf (A)}(z) \left( {\bf \Psi}_{(2)}^{\infty}(z) \right)^{-1} .
\end{equation}
Here we define the notation
\begin{align}\nonumber
{\bf V}^{\bf (A)}&:=(\text{Jump of ${\bf \Psi}^\infty_{(0)}$})^{-1}(\text{Jump of ${\bf W}$})
\\&=\left\{\begin{array}{l}
\begin{array}{ccc}\vspace{0.3cm}
\bpm 1&0&0\\0&1&-e^{n{\cal P}_3(z)}\\0&0&1\epm,&
\bpm 1 & 0 & 0 \\ e^{-n{\cal P}_1(z)} & 1 & 0 \\ 0 & 0 & 1 \epm,&
\bpm 0&e^{n({\rm Re}{\cal P}_1(z)+i\sigma)}&0\\0&1&0\\0&0&1\epm,
\\ 
L=\partial\Omega_\text{out}\cap\partial\Omega_L,&
\partial\Omega_\text{lens}\cap\partial\Omega_L,&
\partial\Omega_L^+\cap\partial\Omega_L^-,
\end{array} 
\\ \\
\begin{array}{c}\vspace{0.3cm}
\bpm 1 & e^{n{\cal P}_1(z)} & e^{n{\cal P}_2(z)} \\ 0 & 1 & 0 \\ 0 & 0 & 1 \epm \\
\partial\Omega_\text{out}^+\cap\partial\Omega_\text{out}^-
\end{array} .
\end{array}\right. 
\end{align}

\item For $z$ on the boundary of $\mathbb{D}_c$, the jump of ${\bf E}_{(2)}$ is given in formula (\ref{v2-jump-dc}).
\item For $z$ on the boundary of $\mathbb{D}_{\mbox{Ai}}$, the jump of ${\bf E}_{(2)}$ is
\begin{equation} \label{crit-135}
{\bf V}_{(2)}^{\bf (E)}(z) = {\bf R}(z) {\bf H}_0(z) {\bf S}(\zeta) \left( {\bf P}_{\mbox{Ai}}(z) \right)^{-1} {\bf R}(z)^{-1} .
\end{equation}
\item On the bands outside of the disks $\mathbb{D}_{\mbox{Ai}}$ and $\mathbb{D}_c$, $ {\bf V}_{(2)}^{\bf (E)}(z) \equiv {\bf I}$.  Also, on contours inside $\mathbb{D}_{\mbox{Ai}}$ and $\mathbb{D}_c$, $ {\bf V}_{(2)}^{\bf (E)}(z) \equiv {\bf I}$.
\end{itemize}

The next result follows from the definition of $\mathfrak{g}(z;\kappa)$ in 
\eqref{mathfrak-g}.
\begin{lemma} \label{crit-7.2}
For sufficiently small $\kappa$ we have the estimate below, uniformly in $z$ over compact sets bounded away from the turning points  $\alpha_j$, $\beta_j$, 
$\alpha$, and $\beta$:
\begin{equation}
\mathfrak{g}(z; \kappa) = g(z) + \mathcal{O}\left(\kappa\right).
\end{equation}
\end{lemma}

\begin{lemma} \label{crit-7.3}
In the critical regime, the inner and outer lenses have been chosen so that 
\begin{enumerate}[(a)]

\item On 
{$L$} outside of the disks $\mathbb{D}_{\mbox{Ai}}$ and $\mathbb{D}(\beta,\delta)$:  The real part of $P_3$ is 
{negative} and bounded away from zero.

\item On 
{$\partial\Omega_{\text{lens}}\cap\partial\Omega_L$} outside of the disks $\mathbb{D}_{\mbox{Ai}}$ and $\mathbb{D}(\beta, \delta)$:  The real part of $P_1$ is positive and bounded away from zero.

\item On the real axis, 
away from the bands and outside of the disks $\mathbb{D}_{\mbox{Ai}}$ and $\mathbb{D}(\beta, \delta)$:
The real part of $P_1$ is negative and bounded away from zero.

\item On the real axis, to the right of $\beta$, and outside of the disk $\mathbb{D}(\beta, \delta)$:  The real part of $P_2$ is negative and bounded away from zero.

\end{enumerate}
\end{lemma}

\begin{proof}
Statements (b) and (c) follow from the analysis of the Riemann-Hilbert problem for the standard orthogonal polynomials (see, for instance, \cite{Deift:1998-book}).  Statement (d) follows from the definition of the critical regime, that $P_2(\beta)$ is a unique global maximum for $P_2(x)$ in $\mathbb{R} \setminus [\alpha, \beta]$ (see Definition \ref{critical}).    
To prove (a), recall that  $ P_3(z) = a z- g(z)- \ell_1 + \ell_2$, and that $P_3(\beta) = 0$.  Along the contour $L$ it is clear that $\mbox{Re}\left[ 
  a z \right] $ is 
  {de}creasing.  Likewise 
\begin{equation}
\mbox{Re}\left[ g(z) \right] = \int_{\mathbb{R}} \log|z-s| \rho_{\mbox{\scriptsize min}}(s) ds 
\end{equation}
increases along $L$, as $|z-s|$ is increasing along this contour for each 
$s \in \text{supp}(\rho_{\mbox{\scriptsize min}})$.  See Lemma \ref{lem_steepest}.  
\end{proof}

We will use the following data about the functions $\mathcal{P}_1(z)$, $\mathcal{P}_2(z)$, and $\mathcal{P}_3(z) $ to control the jumps of the error matrices on the contours outside of the disks.
\begin{lemma} \label{crit-lem-1}
For $\kappa$ sufficiently small:
\begin{enumerate}[(a)]

\item 
There exists $b>0$ such that on 
{$L$} outside of the disk $\mathbb{D}_c$:
\begin{equation}
\mbox{Re}\left[ \mathcal{P}_3\right] 
{<-} b n^{-3/4} .
\end{equation}

\item There exists $b>0$ such that on 
{$\partial\Omega_{\text{lens}}\cap\partial\Omega_L$}  outside of the disks $\mathbb{D}_c$ and $\mathbb{D}_{\mbox{Ai}}$: 
\begin{equation}
\mbox{Re}\left[ \mathcal{P}_1\right] > b n^{-3/4} .
\end{equation}

\item  On the real axis, 
off the bands and outside ${\mathbb D}(\beta,\delta)$ and ${\mathbb D}_{\mbox{Ai}}$:  The real part of $\mathcal{P}_1$ is negative and bounded away from zero. 

\item 
There exists $b>0$ such that on the real axis, to the right of $\beta$, and outside of the disk $\mathbb{D}_c$:
\begin{equation} \label{crit-2c}
 \mbox{Re}\left[ \mathcal{P}_1 \right] < -b n^{-3/4} .
\end{equation}

\item 
There exists a $b>0$ such that on the real axis, to the right of $\beta$, and outside of the disk $\mathbb{D}_c$:
\begin{equation} \label{crit-4c}
\mbox{Re}\left[ \mathcal{P}_2\right] < - b n^{-3/4} .
\end{equation}

\end{enumerate}
\end{lemma}

\begin{proof}

Part (c) follows directly from Lemmas \ref{crit-7.2} and \ref{crit-7.3} (c).  
To prove (a) we divide the contour into two parts:  one is the section of 
{$L$} outside of $\mathbb{D}(\beta, \delta)$ and the other is the section of 
{$L$} inside of $\mathbb{D}(\beta, \delta) \setminus \mathbb{D}_c$.  We have choosen $L$ such that $\mbox{Re}\left( \mathcal{P}_3(z) 
{-} \frac{3\kappa}{2} \log(z-\beta)\right) $ is 
{de}creasing along $L$ (see Lemma \ref{lem_steepest}).
The proof of Lemma \ref{lem_steepest} can be modified to show that $\mbox{Re}\left( P_3(z) \right) $ is also 
{de}creasing along $L$.  When $z\notin \mathbb{D}(\beta, \delta)$, by Lemma \ref{crit-7.2} $\mbox{Re}\left( \mathcal{P}_3(z) 
{-} \frac{3\kappa}{2} \log(z-\beta)\right) $ converges to $P_3(z)$ for $\kappa \to 0$, hence by Lemma \ref{crit-7.3} (a) there is a $\kappa$ small enough so that $\mbox{Re}\left( \mathcal{P}_3(z) \right)$ is 
{de}creasing along $L$ and is therefore 
{negative} on $L \setminus \mathbb{D}_c$, and bounded away from zero on $L \setminus \mathbb{D}(\beta, \delta)$.  

From (\ref{hdef}) one can see that 
\begin{equation}
\mathcal{P}_3(z) = -\frac{1}{2} \mathcal{P}_1(z) + h(z)+ \frac{3\kappa}{2} \log\zeta = h(\beta)+ \frac{2}{3} c_1^{3/2} (z-\beta)^{3/2} +\frac{\tau}{n^{1/3}} c_1 (z-\beta) +\frac{3\kappa}{2} \log\zeta + \mathcal{O}\left( (z-\beta)^2\right).
\end{equation}
We start by noting that $
 h(\beta)
 {<} 0 $ for $n$ sufficiently large so this term will only improve our bound.
Next for $z\in L \cap (\mathbb{D}(\beta, \delta) \setminus \mathbb{D}_c) $, we will show that, for $n$ large enough, the 
$\frac{2}{3} c_1^{3/2}(z-\beta)^{3/2}$ term dominates the other two.  
We have  
\begin{equation}
\left| \frac{ -\frac{\tau}{n^{1/3}} c_1 (z-\beta) }{ - \frac{2}{3} c_1^{3/2} (z-\beta)^{3/2} }\right| = 
\frac{3 \tau}{2 c_1^{1/2}} \frac{1}{n^{1/3} |z-\beta|^{1/2}} 
\leq \frac{3\tau}{2 c_1^{1/2}} \frac{1}{n^{1/12}},
\end{equation}
and
\begin{equation}
\left| \frac{ - \frac{3\kappa}{2} \log\zeta }{-\frac{2}{3} c_1^{3/2} (z-\beta)^{3/2} } \right| = 
\frac{9\kappa}{4c_1^{3/2} } \frac{ |\log\zeta|}{|z-\beta|^{3/2}} 
\leq \frac{3c}{8 c_1^{3/2}} n^{\gamma-\frac{1}{4}} \left[ \log(n) + \mathcal{O}\left( n^{-1/2}\right) \right]
\end{equation}
for some constant $c>0$.

To conclude the proof of (a) we note that there is a $b>0$ such that $\mbox{Re}\left[
\frac{2}{3} c_1^{3/2} (z-\beta)^{3/2} \right] 
{<} b n^{-3/4} $ provided that the segments of $L$ lie in the sector $\pi/3 < \theta < 5 \pi / 3$.

To prove (b) we divide the contour into two parts:  one is the sections of the inner lenses outside of $\mathbb{D}_{\mbox{Ai}}$ and $\mathbb{D}(\beta, \delta)$, the other is the sections of the inner lenses inside of $\mathbb{D}(\beta,\delta)\setminus \mathbb{D}_c $.  That $\mbox{Re}\left[ \mathcal{P}_1\right] $ is positive and bounded away from zero follows from Lemmas \ref{crit-7.2} and \ref{crit-7.3}(b).  Inside $\mathbb{D}(\beta, \delta) \setminus \mathbb{D}_c$ we have 
\begin{equation} \label{crit-7.3-bform}
\mathcal{P}_1(z) = -\frac{4}{3} c_1^{3/2} (z-\beta)^{3/2} + \mathcal{O}\left( (z-\beta)^2\right), 
\end{equation}
and the result follows as above.

To prove (d) and (e) we divide the contour into two parts:  one is the interval $[\beta+\delta, \infty)$ the other is $[\beta+ n^{-1/2}, \beta+\delta)$.  
That $\mbox{Re}\left[ \mathcal{P}_1\right] $ and $\mbox{Re}\left[ \mathcal{P}_2\right] $ are negative and bounded away from zero on the first interval follows from Lemmas \ref{crit-7.2} and \ref{crit-7.3}(c)--(d).  
Within the interval $[\beta+n^{-1/2}, \beta+\delta)$, (\ref{crit-7.3-bform}) gives the result for (d), and for (e) we have 
\begin{equation}
\mathcal{P}_2(z) =  h(\beta) -\frac{2}{3} c_1^{3/2} (z-\beta)^{3/2} + \frac{\tau}{n^{1/3}} c_1 (z-\beta) + \frac{3\kappa}{2}\log\zeta  + \mathcal{O}\left( (z-\beta)^2\right).
\end{equation}
Again, $h(\beta) < 0 $ for $n$ sufficiently large so this term improves the 
bound.
For $z \in [\beta+n^{-1/2}, \beta+\delta)$, we will show that for $n$ large enough the $-\frac{2}{3} c_1^{3/2} (z-\beta)^{3/2}$ term dominates the other two.  We have 
\begin{equation}
\left| \frac{\frac{\tau}{n^{1/3}} c_1 (z-\beta)}{-\frac{2}{3} c_1^{3/2} (z-\beta)^{3/2}} \right| = 
\frac{3\tau}{2 c_1^{1/2}} \frac{1}{n^{1/3}|z-\beta|^{1/2}} \leq \frac{3\tau}{2 c_1^{1/2}} \frac{1}{n^{1/12}},
\end{equation}
and 
\begin{equation}
\left| \frac{\frac{3\kappa}{2}\log\zeta}{-\frac{2}{3} c_1^{3/2} (z-\beta)^{3/2}} 
\right| = 
\frac{9\kappa}{4 c_1^{3/2}} \frac{ | \log\zeta | }{|z-\beta|^{3/2} }
\leq 
\frac{3c}{8 c_1^{3/2}} n^{\gamma-\frac{1}{4}} \left[ \log(n) + \mathcal{O}\left( n^{-1/2} \right) \right]
\end{equation}
for some constant $c>0$.  
To conclude the proof of (e) we note that there is a $b > 0 $ such that 
\newline
$\ds\mbox{Re}\left[ -\frac{2}{3} c_1^{3/2} (z-\beta)^{3/2} \right] < - b n^{-3/4}.$
\end{proof}


We can now give bounds on the jumps ${\bf V}_{(2)}^{\bf (E)}(z)$ of the error problem.
\begin{lemma} \label{crit-error-lem}
In the near-critical regime, for large $n$, 
\begin{enumerate}[(a)]

\item 
Off the boundaries of $\mathbb{D}_c$ and $\mathbb{D}_{\rm{Ai}}$: 
There is a constant $b > 0$ such that 
\begin{equation}
{\bf V}_{(2)}^{\bf (E)}(z) = {\bf I} + \mathcal{O}\left( e^{-bn^{1/4}} \right), \quad 
z \notin (\partial\mathbb{D}_c \cup \partial\mathbb{D}_{\rm{Ai}}).
\end{equation}

\item On the boundary of $\mathbb{D}_c$:
\eq
\label{VE2-bound}
{\bf V}_{(2)}^{\bf(E)}(z) = {\bf I} + \mathcal{O}\left(\frac{1}{n^{2/3}} \right), \quad z\in\partial\mathbb{D}_c.
\endeq

\item On the boundary of $\mathbb{D}_{\rm{Ai}}$:
\eq
\label{V2E-Dalpha}
{\bf V}_{(2)}^{\bf(E)}(z) = {\bf I} + \mathcal{O}\left(\frac{1}{n} \right), \quad z\in\partial\mathbb{D}_{\rm{Ai}}.
\endeq

\end{enumerate}
\end{lemma}

\begin{proof}

Part (a) follows from equation (\ref{crit-version-63}), Lemma \ref{crit-lem-1}, and the boundedness of ${\bf \Psi}_{(2)}(z)$.

Part (b) follows from \eqref{v2-jump-dc} along with \eqref{L2-bound}, 
\eqref{H2-n-decay}, and the uniform boundedness of ${\bf D}(z)$ inside 
$\mathbb{D}_c$ as $n\to\infty$.

For part (c), first recall from the Schlesinger calculations that 
${\bf F}_{(1,m)}=\mathcal{O}(r/n^{1/2})$ and 
${\bf F}_{(2,m)}=\mathcal{O}(r/n^{1/3})$ (as follows from 
\eqref{Fj-bound}, \eqref{H1bound}, and \eqref{F2m}).  Recalling the definition of ${\bf R}(z)$ in 
\eqref{R-def}, we see
that
\eq
{\bf R}(z;n) = {\bf I} + \mathcal{O}\left(\frac{r}{n^{1/3}}\right).
\endeq
Now from (\ref{crit-135}) we have
\begin{equation}
{\bf V}_{(2)}^{\bf (E)}(z) = {\bf R}(z) \left( {\bf I} + \mathcal{O}\left(\frac{1}{n}\right) \right) {\bf R}(z)^{-1},
\end{equation}
from which \eqref{V2E-Dalpha} follows.
\end{proof}

\begin{lemma}
\label{E2-error}
In the critical regime, for $n$ large,
\[ {\bf E}_{(2)}(z) = {\bf I} + \mathcal{O}\left( \frac{1}{n^{2/3}} \right) \]
uniformly in $z$. 
\end{lemma}

\begin{proof}
Denote $\Gamma_C := \partial \mathbb{D}_c \cup \partial \mathbb{D}_{\mbox{Ai}} $ and 
$\Gamma_N$ the remaining contours on which the error $E_{(2)}(z)$ has a jump: the inner and outer lenses, and the real axis, outside of the regions $\mathbb{D}_c$ and $\mathbb{D}_{\mbox{Ai}}$.
From Lemma \ref{crit-error-lem}(b)--(c), 
\begin{equation}
{\bf V}_{(2)}^{({\bf E})}(z) = {\bf I} + \mathcal{O}\left( \frac{1}{n^{2/3}}\right), \quad z \in \Gamma_C.
\end{equation}
Then for $n$ sufficiently large there exists a constant $c$ such that 
\begin{equation}
|| {\bf V}_{(2)}^{({\bf E})} - {\bf I} ||_{L^2(\Gamma_C)} + || {\bf V}_{(2)}^{({\bf E})} - {\bf I} ||_{L^\infty(\Gamma_C )} \leq c n^{\insrt{-}2/3}.
\end{equation}
From Lemma \ref{crit-error-lem}(a), for $n$ sufficiently large there is a constant $c$ such that 
\begin{equation}
|| {\bf V}_{(2)}^{({\bf E})} - {\bf I} ||_{L^2(\Gamma_N)} + || {\bf V}_{(2)}^{({\bf E})} - {\bf I} ||_{L^\infty(\Gamma_N )} \leq c e^{-cn} .
\end{equation}
The lemma then follows by a standard technique that consists of writing the solution to the Riemann-Hilbert problem in terms of a Neumann series involving ${\bf V}^{({\bf E})} - {\bf I} $ ( see, for instance, \cite{Deift:1999b} Section 7.2 or \cite{Ercolani:2003} Section 3.5). 
\end{proof}


\section{The kernel near the critical region}
\label{kernel-section}

Our main goal now is to obtain the asymptotic form of the kernel 
$K_n(x(\zeta_x),y(\zeta_y))$ (see \eqref{mop-kernel}) uniformly
for 
$\zeta_x,\zeta_y$  in compact subsets. 
The first observation is that $K_{n}(x,y)$ is smooth in $x$ and $y$ because i) the first column of ${\bf Y}$ has no jump and ii) the second and the third rows of ${\bf Y}^{-1}$ have no jump.   However, it is still possible for the 
leading-order asymptotics of the kernel to exhibit a Stokes--like phenomenon, 
namely, a discontinuous change with respect to parameters.
We now show that the leading asymptotics of $K_{n}$ are also smooth.

Combining \eqref{W}, \eqref{E2}, \eqref{crit-psi-2}, and \eqref{L0-of-zeta}, 
we see we can write ${\bf Y}$ for $z\in\mathbb{D}_c$ as
\begin{equation}
\label{Y-expression1}
{\bf Y}= \sqrt\pi\left(\frac{z-\beta}\zeta\right)^{r/2} {\bf \Lambda}{\bf E}_{(2)}\, {\bf H}_{(2)}{\bf D}^{-1}\left(\begin{array}{ccc}n^{1/6} & 0 & 0 \\0 & n^{-1/6} & 0 \\0 & 0 & 1\end{array}\right){\cal A}_{r} 
 {\bf J}^{-1}\left(\begin{array}{ccc}e^{-\frac{n}{2}V} & 0 & 0 \\0 & e^{\frac{n}{2}V} & 0 \\0 & 0 & e^{\frac{n}{2}(V-2a z)}\end{array}\right)^{-1}.
\end{equation}
We note that ${\cal A}_{r}{\bf J}^{-1}$ and 
${\bf Y}\times\text{diag}[e^{-\frac{n}{2}V}, e^{\frac{n}{2}V} , e^{\frac{n}{2}(V-2a z)}]$ are the same up to a holomorphic prefactor.  Therefore they 
have the same jump, which is
\begin{equation}
\left({\cal A}_{r}{\bf J}^{-1}\right)_{+}=\left({\cal A}_{r}{\bf J}^{-1}\right)_{-}\left(\begin{array}{ccc}1 & 1 & 1 \\0 & 1 & 0 \\0 & 0 & 1\end{array}\right), \quad {\zeta\in\mathbb R}.
\end{equation}
This means that the first column of ${\cal A}_{r}{\bf J}^{-1}$ and the 
second and third rows of ${\bf J}{\cal A}_{r}^{-1}$ have no jump.  
Since the leading term in the asymptotic expansion of the kernel is written 
in terms of those column and rows, our asymptotic expression of the kernel 
is smooth.   Therefore, it is enough to consider, say, $\zeta$ in region I 
(see Figure \ref{fig_critloc}).

To arrive at our final expression for the kernel we will need to express 
$\mathcal{A}_r(\zeta_y)^{-1}\mathcal{A}_r(\zeta_x)$ in a simple form 
involving contour integrals.  To 
do so we take advantage of the machinery of the \emph{bilinear concomitant}.  
This requires writing $\mathcal{A}_r(\zeta)$ as a constant multiple of a 
Wronskian matrix.  Specifically, in region I,
\eq
\label{Ar-ito-W}
\mathcal{A}_r(\zeta) = \bpm \tau & 1 & 0 \\ 0 & \tau & 1 \\ 1 & 0 & 0 \epm {\bs \chi}_{r-1}(\zeta), \quad \zeta\in\text{I},
\endeq
where
\eq
{\bs \chi}_r(\zeta):=\bpm {\rm Ai}^{(r)}_{{\cal C}_1}(\zeta) & {\rm Ai}^{(r)}_{{\cal C}_2}(\zeta)& {\rm Ai}^{(r)}_{{\cal C}_3}(\zeta) \\ \partial_\zeta {\rm Ai}^{(r)}_{{\cal C}_1}(\zeta) &\partial_\zeta {\rm Ai}^{(r)}_{{\cal C}_2}(\zeta)&\partial_\zeta {\rm Ai}^{(r)}_{{\cal C}_3}(\zeta) \\ \partial_\zeta^2{\rm Ai}^{(r)}_{{\cal C}_1}(\zeta) &\partial_\zeta^2{\rm Ai}^{(r)}_{{\cal C}_2}(\zeta) & \partial_\zeta^2{\rm Ai}^{(r)}_{{\cal C}_3}(\zeta)\epm.
\endeq
Eq. \eqref{Ar-ito-W} is obtained from the identity
\begin{equation}
\partial_{\zeta}{\rm Ai}_{{\cal C}_j}^{(r-1)}(\zeta)={\rm Ai}_{{\cal C}_j}^{(r)}(\zeta)-\tau {\rm Ai}_{{\cal C}_j}^{(r-1)}(\zeta).
\end{equation}
From \eqref{Y-expression1} and \eqref{Ar-ito-W}, we can write 
\begin{equation}
\begin{split}
\label{Y-expression2}
{\bf Y}(z)&= \sqrt\pi\left(\frac{z-\beta}\zeta\right)^{r/2} {\bf \Lambda}{\bf E}_{(2)}{\bf H}_{(2)}{\bf D}^{-1} \left(\begin{array}{ccc}n^{1/6} & 0 & 0 \\0 & n^{-1/6} & 0 \\0 & 0 & 1\end{array}\right) \times \\  
& \qquad \qquad \times \left(\begin{array}{ccc}\tau & 1 & 0 \\0 & \tau & 1 \\1 & 0 & 0\end{array}\right)
{\bs \chi}_{r-1}(\zeta) \left(\begin{array}{ccc}e^{-\frac{n}{2}V} & 0 & 0 \\0 & e^{\frac{n}{2}V} & 0 \\0 & 0 & e^{\frac{n}{2}(V-2a z)}\end{array}\right)^{-1}, \quad \zeta\in\text{I}.
 \end{split}
\end{equation}
We will now express ${\bs \chi}_{r-1}(\zeta_y)^{-1}{\bs \chi}_{r-1}(\zeta_x)$ in a simple 
form using the bilinear concomitant.

\subsection{Simplifying 
\texorpdfstring{$\bs{\chi_{r-1}(\zeta_y)^{-1}\chi_{r-1}(\zeta_x)}$}{21932} using the bilinear concomitant}

This proposition is a specific example of the more general results on the bilinear concomitant given in \cite{Bertola:2007a}.  
\begin{prop}
\label{prop-4-15}
Recall the contours $\mathcal{C}_1$, $\mathcal{C}_2$, and $\mathcal{C}_3$ are 
defined in Figure \ref{fig_airycontour}.  Let the dual contours 
$\widehat{\mathcal{C}}_1$, $\widehat{\mathcal{C}}_2$, and $\widehat{\mathcal{C}}_3$ 
be defined as in Figure \ref{alt-k-airy-contours}.  Then the entries of 
${\bs \chi}_{r-1}(\zeta_y)^{-1} {\bs \chi}_{r-1}(\zeta_x)$ for $i\neq j$ are given by
\be
\le({\bs \chi}_{r-1}(\zeta_y)^{-1} {\bs \chi}_{r-1}(\zeta_x)\ri)_{ij} = \frac{1}{2\pi i} (\zeta_x-\zeta_y) \int_{\wh{\mathcal C}_i} \d s \int_{{\mathcal C}_j} \d t \frac {(t+\tau)^r}{(s+\tau)^r}\frac{ {\rm e}^{\frac {s^3-t^3}3 + \zeta_x t-\zeta_y s}}{t-s}.
\ee
\end{prop}

\begin{proof}

In the proof 
we will use the same notation as \cite{Bertola:2007} Section 3, and the technical details whose proofs we leave out are given there.

Recall
\begin{equation}
{\rm Ai}_{\mathcal{C}_j}^{(r-1)}(\zeta) =
 \frac{1}{2\pi i} \int_{\mathcal{C}_j} (t + \tau)^{r-1} e^{\zeta t - t^3 / 3} dt, \quad j=1,2,3.
\end{equation}
The main observation is that, for fixed $r$,  each ${\rm Ai}_{\mathcal C_j}^{(r-1)}\ (j=1,2,3)$ satisfies the same 
 ordinary differential equation of third order which we now derive using integration by parts: 
\eq
\begin{split}
0 & =\frac{1}{2i\pi}\int_{{\cal C}_j}\frac{d}{dt}\left((t+\tau)^{r}e^{\zeta t-t^3/3}\right)dt =\frac{1}{2i\pi}\int_{{\cal C}_j}\big(r+(\zeta-t^2)(t+\tau)\big)(t+\tau)^{r-1}e^{\zeta t-t^3/3}dt  \\
  & =\frac{1}{2i\pi}\int_{{\cal C}_j}\left(r+(\zeta-\partial_\zeta^2)(\partial_\zeta+\tau)\right)(t+\tau)^{r-1}e^{\zeta(t+\tau)-t^3/3}dt \\
  &= \left(\zeta (\partial_\zeta+\tau)+(r-\partial_\zeta^3-\tau\partial_\zeta^2)\right){\rm Ai}_{{\cal C}_j}^{(r-1)}(\zeta).
\end{split}
\endeq

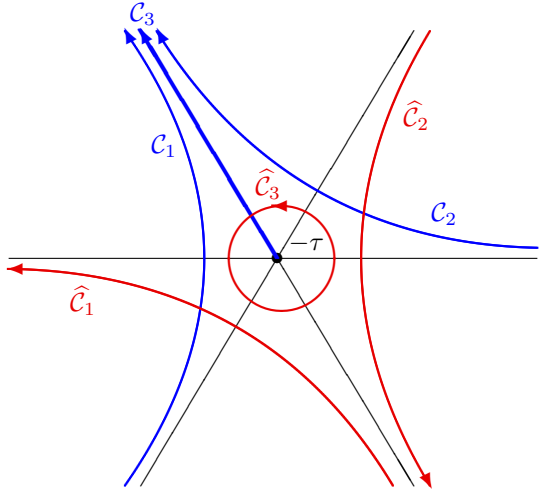
\begin{figure}
\setlength{\unitlength}{2pt}
\begin{center}
\begin{picture}(100,100)(-50,-50)
\put(-50,0){\line(1,0){100}}
\put(-25,43){\line(3,-5){51.7}}
\put(-25,-43){\line(3,5){51.7}}
\put(0.8,0){\circle*{2}}
\put(3,1.5){$-\tau$}
\thicklines
\blue{
\put(-23,20){$\mathcal{C}_1$}
\qbezier(-28,43)(2,0)(-28,-43)
\put(-27.1,41.6){\vector(-1,2){1}}
\put(30,7){$\mathcal{C}_2$}
\qbezier(-22,43)(3,3)(50,2)
\put(-21,41.6){\vector(-1,2){1}}
\put(-27,45){$\mathcal{C}_3$}
\put(0.8,0){\line(-3,5){26}}
\put(0.6,0){\line(-3,5){26}}
\put(1,0){\line(-3,5){26}}
\put(-24.1,41.6){\vector(-1,2){1}}
}
\red{
\put(-40,-10){$\widehat{\mathcal{C}}_1$}
\qbezier(-51,-2)(-3,-3)(21,-43)
\put(-51,-2){\vector(-1,0){1}}
\put(23,25){$\widehat{\mathcal{C}}_2$}
\qbezier(28,43)(2,0)(28,-43)
\put(27.4,-41.6){\vector(1,-2){1}}
\put(-5,12){$\widehat{\mathcal{C}}_3$}
\put(0,0){\circle{20}}
\put(-1,10){\vector(-1,0){1}}
}
\end{picture}
\end{center}
\caption{The choice of contours for the $r$-Airy parametrix in region I. \label{alt-k-airy-contours}}
\end{figure}

Associated to each of these equations is an \emph{adjoint equation} 
\cite {Ince} obtained by replacing $\pa_\zeta\mapsto -\pa_{\zeta}$ and interchanging the order of multiplications and differentiation operators.
Its 
solutions have the form
\be
\wh{\rm Ai}_{\wh{\cal C}_j}^{(r-1)}(\zeta)  := \int_{\wh {\cal C}_j} \frac 1{(s+\tau)^r} {\rm e}^{\frac {s^3}3 - \zeta s}\d s, \quad j=1,2,3,
\ee
where the \emph{dual contours} $\wh{\cal C}_1$, $\wh{\cal C}_2$, and 
$\wh{\cal C}_3$ are defined in Figure 
\ref{alt-k-airy-contours}. 
The bilinear concomitant \insrt{\cite{Ince}} for this differential equation is a bilinear pairing between the solution space of an equation and its adjoint. For the case at hand it admits the following double--integral representation \cite{Bertola:2007a}:
\be \label{Bdef}
\mathcal B\left({\rm Ai}_{{\cal C}_j}^{(r-1)}, \wh{\rm Ai}_{\wh{\cal C}_i}^{(r-1)}\right) = \frac{1}{2\pi i} \int_{\wh {\cal C}_i} \int_{ {\cal C}_j} \le( (t+s)(t+\tau)+s^2-\zeta \ri) \frac {(t+\tau)^{r-1}}{(s+\tau)^r} {\rm e}^{\frac {s^3-t^3}3 + \zeta(t-s)}dt\,ds.
\ee
Although $\zeta$ appears in \eqref{Bdef} it can be seen by direct differentiation that the expression \eqref{Bdef} is {\em independent of $\zeta$}.
We may follow the steps in the proof of Lemma 3.3 in \cite{Bertola:2007a} to show that \footnote{The paper \cite{Bertola:2007} contains a wrong sign in front of the intersection pairing. In fact formulas (3.38) and (3.39) should have the opposite sign in front of the second term in the respective integrands.
Moreover, the wrong overall sign was obtained in using Lemma 3.3 to derive the final formula since formula (3.37) was to be  the difference of (3.39) minus (3.38)  but apparently it was miscomputed later as (3.38) minus (3.39).}
\be \label{Bpairs}
\mathcal B\left({\rm Ai}_{{\cal C}_j}^{(r-1)}, \wh{\rm Ai}_{\wh{\cal C}_i}^{(r-1)}\right) =  {\cal C}_j \sharp \wh {\cal C}_i,
\ee
where ${\cal C}_j\sharp\wh{\cal C}_i$ is the intersection number of the 
contours ${\cal C}_j$ and $\wh{\cal C}_i$.  
Introduce the Wronskian $\wh{\bs \chi}_r(\zeta)$ of solutions to the 
adjoint equations with entries
\be
\left(\wh{\bs \chi}_{r}(\zeta)\right)_{i m}:= (-1)^{ m-1}\partial_\zeta^{ m-1}\wh{\rm Ai}_{\wh{\cal C}_i}^{(r)}(\zeta), \quad 1\leq i, m\leq 3,
\ee
and denote
\begin{equation}
{\bf F}(\zeta) := \begin{pmatrix} \zeta & -\tau & -1 \\ - \tau & -1 & 0 \\ -1 & 0 & 0 \end{pmatrix}.
\end{equation}
Then, by the definition of the bilinear concomitant in \eqref{Bdef} and the 
pairing \eqref{Bpairs},
\be
\left[ \mathcal B \left({\rm Ai}_{{\cal C}_j}^{(r-1)}, \wh{\rm Ai}_{\wh{\cal C}_i}^{(r-1)}\right) \right]_{1 \leq i, j \leq 3}  = 
\wh {\bs \chi}_{r-1}(\zeta) {\bf F}(\zeta) {\bs \chi}_{r-1}(\zeta) = {\bf I}.
\ee
Therefore $ \wh {\bs \chi}_{r-1}{\bf F}$ is the inverse of ${\bs \chi}_{r-1}$.
Expanding out the entries of ${\bs \chi}_{r-1}$ and $\wh {\bs \chi}_{r-1}$ in  
\be 
\le({\bs \chi}_{r-1}(\zeta_y)^{-1} {\bs \chi}_{r-1}(\zeta_x)\ri)_{ij} = 
\sum_{ k,  m} \left(\wh{\bs \chi}_{r-1}(\zeta_y)\right)_{i  k} \left({\bf F}(\zeta_y)\right)_{ k  m} \left({\bs \chi}_{r-1}(\zeta_x)\right)_{ m j}
\ee
gives
\be
\le({\bs \chi}_{r-1}(\zeta_y)^{-1} {\bs \chi}_{r-1}(\zeta_x)\ri)_{ij} = \frac{1}{2\pi i}\int_{\wh {\cal C}_i} \int_{ {\cal C}_j} \le( (t+s)(t+\tau)+s^2-\zeta_y \ri) \frac {(t+\tau)^{r-1}}{(s+\tau)^r} {\rm e}^{\frac {s^3-t^3}3 + \zeta_x t- \zeta_y s}dt\,ds.
\label{617}
\ee
We are interested only in the case $i\neq j$ and hence the contours of integration in \eqref{617} have no intersection;  this allows us to integrate by parts and obtain in the integrand a harmless denominator $(t-s)$.
Then we have the straightforward chain of equalities, where only integration by parts is used:
\eq
\begin{split}
\int_{\wh{\mathcal C}_i } \int_{{\mathcal C}_j} \frac {(t+\tau)^{r-1}}{(s+\tau)^r}&\Big((t+s)(t+\tau)+s^2-\zeta_y \Big) {\rm e}^{\frac {s^3-t^3}3+\zeta_x t-\zeta_y s}dt\,ds \\
  & = \int_{\wh{\mathcal C}_i} \int_{{\mathcal C}_j} \frac {(t+\tau)^{r-1}}{(s+\tau)^r}\Big((t+s)(t+\tau)+\partial_s \Big) {\rm e}^{\frac {s^3-t^3}3+\zeta_x t-\zeta_y s}dt\,ds \\
  & = \int_{\wh{\mathcal C}_i} \int_{{\mathcal C}_j} \left((t+s)\frac {(t+\tau)^r}{(s+\tau)^r}-\left[\partial_s\frac {(t+\tau)^{r-1}}{(s+\tau)^r}\right] \right) {\rm e}^{\frac {s^3-t^3}3+\zeta_x t-\zeta_y s} dt\,ds \\
  & = \int_{\wh{\mathcal C}_i} \int_{{\mathcal C}_j} \left((t+s)+\frac{r}{(t+\tau)(s+\tau)} \right) \frac {(t+\tau)^r}{(s+\tau)^r} {\rm e}^{\frac {s^3-t^3}3+\zeta_x t-\zeta_y s} dt\,ds \\
  & = \int_{\wh{\mathcal C}_i} \int_{{\mathcal C}_j} \frac{{\rm e}^{\zeta_x t-\zeta_y s}}{t-s}\left(\frac{-r}{t+\tau}+t^2+\frac{r}{s+\tau}-s^2\right) \frac {(t+\tau)^r}{(s+\tau)^r} {\rm e}^{\frac {s^3-t^3}3} dt\,ds \\
  & = \int_{\wh{\mathcal C}_i} \int_{{\mathcal C}_j} \frac{{\rm e}^{\zeta_x t-\zeta_y s}}{t-s}(-\pa_t - \pa_s) \frac {(t+\tau)^r}{(s+\tau)^r} {\rm e}^{\frac {s^3-t^3}3} dt\,ds \\
  & =  \int_{\wh{\mathcal C}_i} \int_{{\mathcal C}_j} \frac {(t+\tau)^r}{(s+\tau)^r} {\rm e}^{\frac {s^3-t^3}3}(\pa_t + \pa_s)\frac{{\rm e}^{\zeta_x t-\zeta_y s}}{t-s} dt\,ds \\
  & = (\zeta_x-\zeta_y) \int_{\wh{\mathcal C}_i} \int_{{\mathcal C}_j} \frac {(t+\tau)^r}{(s+\tau)^r}\frac{ {\rm e}^{\frac {s^3-t^3}3 + \zeta_x t-\zeta_y s}}{t-s} dt\,ds,
\end{split} 
\endeq
as desired.
\end{proof}

\subsection{Proof of Theorem \ref{crit-kernel-thm}}
\label{kernel-proof}

\begin{proof}
From \eqref{zeta-ito-zmbeta}, it follows that 
\eq
\label{H2invH2}
{\bf H}_{(2)}(y)^{-1} {\bf H}_{(2)}(x) = {\bf I} + \mathcal{O}(x-y) = {\bf I} + \mathcal{O}\left( \frac{\zeta_x - \zeta_y}{n^{2/3}}\right), \quad x,y\in\mathbb{D}_c.
\endeq
Recalling ${\bf E}_{(2)}(z)={\bf I}+{\cal O}(1/n^{2/3})$ (see Lemma \ref{E2-error}), we see
\eq
\begin{split}
&\left(\begin{array}{ccc}\tau & 1 & 0 \\0 & \tau & 1 \\1 & 0 & 0\end{array}\right)^{-1}
\left(\begin{array}{ccc}n^{1/6} & 0 & 0 \\0 & n^{-1/6} & 0 \\0 & 0 & 1\end{array}\right)^{-1}{\bf H}_{(2)}(x)^{-1}{\bf E}_{(2)}(y)^{-1} \times \\
& \qquad \qquad \qquad \times {\bf E}_{(2)}(x){\bf H}_{(2)}(x)\left(\begin{array}{ccc}n^{1/6} & 0 & 0 \\0 & n^{-1/6} & 0 \\0 & 0 & 1\end{array}\right)\left(\begin{array}{ccc}\tau & 1 & 0 \\0 & \tau & 1 \\1 & 0 & 0\end{array}\right)={\bf I}+{\cal O}\left(\frac{\zeta_{x}-\zeta_{y}}{n^{1/3}}\right).
\end{split}
\endeq
Using this and \eqref{Y-expression2}, the kernel is given by
\eq
\begin{split}
K_{n}(x,y)&=\frac{e^{-\frac n2(V(x)-V(y))}}{2\pi i(x-y)}\left[0,e^{-nV(y)},e^{-n(V(y)-a y)}\right]{\bf Y}(y)^{-1}{\bf Y}(x)\left[\begin{array}{c}1 \\0 \\0\end{array}\right]
\\&=\frac{e^{-\frac n2(V(x)-V(y))}}{2\pi i(x-y)}\left[0,e^{-\frac n2V(y)},e^{-\frac n2V(y)}\right]{\bs \chi}(\zeta_y)^{-1}\left({\bf I}+{\cal O}\left(\frac{\zeta_{x}-\zeta_{y}}{n^{1/3}}\right)\right){\bs \chi}(\zeta_x)\left[\begin{array}{c}e^{\frac n2V(x)} \\0 \\0\end{array}\right]
\\&=\frac{1}{2\pi i(x-y)}\left[0,1,1\right]{\bs \chi}(\zeta_y)^{-1}{\bs \chi}(\zeta_x)\left[\begin{array}{c}1 \\0 \\0\end{array}\right]\left(1+{\cal O}\left(\frac{\zeta_{x}-\zeta_{y}}{n^{1/3}}\right)\right)
\\&=\frac{ (\zeta_x - \zeta_y)}{(2\pi i)^2 (x-y)}
\int_{\hat{\mathcal{C}}_2 + \hat{\mathcal{C}}_3} \int_{\mathcal{C}_1} 
\frac{ (t+\tau)^r}{(s+\tau)^r} \frac{e^{\frac{1}{3} (s^3 - t^3) + \zeta_x t - \zeta_y s }}{t-s} dt\,ds \left(1+{\cal O}\left(\frac{\zeta_{x}-\zeta_{y}}{n^{1/3}}\right)\right),
\label{624}
\end{split}
\endeq
where, in the last equality, we have used Proposition \ref{prop-4-15}.
We now use 
\begin{equation}
\frac{\zeta_{x}-\zeta_{y}}{x-y}=n^{2/3}c_{1}(\kappa)\left(1+{\cal O}(n^{-2/3})\right), 
\end{equation}
which follows from \eqref{zeta-ito-zmbeta}. 

We now recall that 
\begin{equation}\label{625}\begin{split}
\zeta_z(\kappa)=\zeta(z;\kappa) &=n^{\frac 23}  c_1(\kappa) (z-\beta(\kappa))(1 + \mathcal O(z-\beta(\kappa)))
\\&=n^{2/3} c_1(0) (z-\beta(0)-\kappa\dot\beta)+{\cal O}(n^{-2/3})
\end{split}
\end{equation}
as long as $z-\beta(0)={\cal O}(n^{-2/3})$, i.e. as long as $\zeta_z$ stays finite for $n\to\infty$.
\\
This implies that we can replace $\zeta_x$ and $\zeta_y$ in the kernel by the leading approximation in \eqref{625} without changing the error in the kernel.
The mismatch from this approximation produces an error of smaller order than the one already indicated in  \eqref{624} and yields the overall error term $\mathcal O(n^{-1/3})$ advocated in Theorem \ref{crit-kernel-thm}.


By dropping the $\mathcal O(n^{-2/3})$ in \eqref{625}
and using $\delta:=c_1\dot\beta\kappa n^{2/3}$ in Theorem \ref{crit-kernel-thm},
\begin{equation}
\zeta_{z} := n^{2/3}c_1(0)(z-\beta)-\delta\quad\Longleftrightarrow \quad z=\beta+\frac{\zeta_z+\delta}{n^{2/3}c_1(0)} ,
\end{equation}
 and evaluating $K(x,y)$ at 
\be
x \equiv  \beta + \frac {\zeta_x+\delta}{n^\frac 23 c_1(0)}\ ,\ \ \ 
y \equiv  \beta + \frac {\zeta_y+\delta}{n^\frac 23 c_1(0)}\ ,
\ee
we obtain the statement in Theorem \ref{crit-kernel-thm}.

We conclude the proof noting that $\mathcal{C}_1$ may be deformed to $\mathcal{C}$ and that $\hat{\mathcal{C}}_2 + \hat{\mathcal{C}}_3$ may be deformed to $\widetilde{\mathcal{C}}$.
\end{proof}

\subsection{Connection to the kernel in \cite{Adler:2009a} for nonintersecting Brownian motions}
\label{brownian-motions}

In conclusion
 we observe that our expression for the kernel 
near the critical point is the same 
that 
was
found previously in the 
literature for nonintersecting Brownian motions:  Theorem 0.1 of 
\cite{Adler:2009a} gives the formula
\begin{equation} \label{Adler:2009-kernel}
 K_n(\zeta_x, \zeta_y) = \frac{1}{4\pi^2}
 \int_0^\infty \int_C \int_C 
 \frac{ (-ib -\tau)^r }{ ( ia - \tau)^r } e^{\frac{1}{3} ( i a^3 + i b^3 ) + i a ( w + \zeta_x) + i b ( w + \zeta_y) } db\,da\,dw ,
 \end{equation}
 where $C$ is a contour proceeding from $\infty e^{5\pi i / 6}$ to $\infty e^{\pi i / 6} $ and such that $-i \tau$ is above $C$.
To see that our formula matches \eqref{Adler:2009-kernel} we first 
 compute the $w$-integral to find
\begin{equation}
 K_n(\zeta_x, \zeta_y) = \frac{1}{(2 \pi i)^2}
\int_C \int_C 
\frac{ (-ib -\tau)^r }{(ia -\tau)^r} \frac{ e^{ \frac{1}{3} ( i a^3 + i b^3) + i a \zeta_x + i b  \zeta_y } }{ i a + i b } db\,da.
\end{equation}
One makes the changes of variables $ i b \mapsto t$ and $ -ia \mapsto  s$.
\begin{equation}
K_n(\zeta_x, \zeta_y) = \frac{1}{(2 \pi i)^2}
\int_{ - i C } \int_{ i C}  \frac{ (t+\tau)^r}{(s+\tau)^r} \frac{ e^{\frac{1}{3} (s^3 - t^3) + \zeta_x t - \zeta_y s }}{ t- s} dt\,ds.
\end{equation}
We conclude by observing
 that $i C = \mathcal{C}$ (i.e. $C$ rotated counter-clockwise by 90 degrees) and $-i C = \widetilde{\mathcal{C}} $ (i.e. $C$ rotated clockwise by 90 degrees) (with $\mathcal C, \widetilde {\mathcal C}$ depicted in Figure \ref{AiryKernelContours}).  Note that the $dt$ contour can be deformed (in the finite complex plane) as the integrand lacks a pole in $t$-space.

\begin{appendix}

\section{The asymptotic expansion of 
\texorpdfstring{$\bs{\mathcal{A}_{(r)}(\zeta)}$}{3252}}
\label{proofs-appendix}

\subsection{Proof of Proposition \ref{Ar-col12}}
Here we prove the existence of general asymptotic expansions for the first 
and second columns of $\mathcal{A}_{(r)}(\zeta)$ with arbitrarily small 
errors (i.e. (\ref{arbitrary-v1})--(\ref{arbitrary-32})).  
First, observe that \eqref{arbitrary-21} and \eqref{arbitrary-31} follow from 
\eqref{arbitrary-v1}, while \eqref{arbitrary-22} and \eqref{arbitrary-32} follow from 
\eqref{arbitrary-v2}.  We then perform a steepest-descent analysis on the original integral representation of the generalized Airy function
\begin{align}\label{Aiexpo}
&{\rm Ai}_{{\cal C}_j}^{(r)}(\zeta):=\frac{1}{2i\pi}\int_{{\cal C}_j}(t+\tau)^r e^{\zeta t-t^3/3}dt=\frac{1}{2i\pi}\int_{{\cal C}_j}e^{r\log\left(t+\tau\right)+\zeta t-\frac{t^3}3}dt .
\end{align}
Let us define $Z$, ${\cal T}$, and $T$ to be scaled versions of $\zeta$, 
$\tau$, and $t$, respectively:
\begin{equation}
Z:=\frac{\zeta}{r^{2/3}},\quad {\cal T}:=\frac{\tau}{r^{1/3}},\quad T:=\frac{t}{r^{1/3}}.
\end{equation}
Using these new variables, the exponent in \eqref{Aiexpo} becomes 
\begin{equation}
r\log\left(t+\tau\right)+\zeta t-\frac{t^3}3=\frac{r}{3}\log r+r F(T),\quad F(T)=F(T;\{Z,{\cal T}\}):=\log(T+{\cal T})+Z T-\frac{T^3}{3}.
\end{equation}
We define the saddle points $T_{\bf s}$ for ${\bf s}=1,2,3$, as the solutions of $\partial_T F(T)=0$.   As solutions to a cubic equation, the $T_{\bf s}$'s can be written explicitly.  Let us, however, only write their asymptotic expressions at large $Z$:
\begin{align}\label{t1t2t3}
\begin{split}
T_1 &=-\sqrt{Z}+\frac{1}{2 Z}+\frac{{\cal T}}{2Z^{3/2}}+\frac{{\cal T}^2}{2 Z^2}+{\cal O}\left(\frac{1}{Z^{5/2}}\right) \\
&= -\sqrt{Z} + \sum_{j=3}^{3m-1} r_j^{(1)} Z^{(1-j)/2} + {\cal O}\left(\frac{1}{Z^{3m/2}}, \frac{\cal T}{Z^{(3m+1)/2}}, \frac{{\cal T}^2}{Z^{(3m+2)/2}}\right),
\\T_2&=\sqrt{Z}+\frac{1}{2 Z}-\frac{{\cal T}}{2Z^{3/2}}+\frac{{\cal T}^2}{2   Z^2}+{\cal O}\left(\frac{1}{Z^{5/2}}\right) \\
&=\sqrt{Z} + \sum_{j=3}^{3m-1} r_j^{(2)} Z^{(1-j)/2} + {\cal O}\left(\frac{1}{Z^{3m/2}}, \frac{\cal T}{Z^{(3m+1)/2}}, \frac{{\cal T}^2}{Z^{(3m+2)/2}}\right),
\\&T_3=-{\cal T}-\frac{1}{Z}-\frac{{\cal T}^2}{Z^2}+{\cal O}\left(\frac{1}{Z^{5/2}}\right).
\end{split}
\end{align}
For large $r$ the coefficients have the asymptotic form $r_j^{(i)} = {\cal O}\left( {\cal T}^{j \mathrm{mod}(3)} \right)$.  This can be seen from the cubic structure of the equation $\partial_T F(T) = 0$. 
It should be mentioned that the above bounds are uniform in ${\cal T}$ inside a finite disk around ${\cal T}=0$.
 
Given a saddle point $T_{\bf s}$, one can expand $F(T)$ by
\begin{equation}
F(T)=F(T_{\bf s})+\frac{F''(T_{\bf s})}{2}(T-T_{\bf s})^2+\delta F(T),\quad \delta F(T)={\cal O}((T-T_{\bf s})^3).
\end{equation}

The standard steepest descent method gives the leading behavior (in large $r$) by
\begin{equation}\label{AiLead}
{\rm Ai}_{{\cal C}_j}^{(r)}(r^{2/3} Z;r^{1/3}{\cal T})\sim\frac{r^{(r+1)/3}}{2\pi{\rm i}}{\rm e}^{rF(T_{\bf s})}\int {\rm e}^{\frac{r}{2} F''(T_{\bf s})(T-T_{\bf s})^2}dT=\frac{r^{(r+1)/3}}{2\pi{\rm i}}\frac{{\rm e}^{rF(T_{\bf s})}}{\sqrt{-\frac{r}{2} F''(T_{\bf s})}} \left(\pm\int_{-\infty}^{\infty}\!\!\!\!{\rm e}^{-x^2}dx\right),
\end{equation}
where the overall sign must be determined from the direction of the contour ${\cal C}_j$.

We find the expansions 
\begin{lemma} \label{exp-F}
For $s=1,2$:
\begin{align}
F(T_s) &= (-1)^{s} \frac{2}{3} Z^{3/2} + \log\left( (-1)^s \sqrt{Z} + {\cal T}\right) + \sum_{j=3}^{3m-1} F_{0, j}^{(s)} Z^{-j/2} + {\cal O}\left(\frac{1}{Z^{3m/2}}, \frac{\cal T}{Z^{(3m+1)/2}}, \frac{{\cal T}^2}{Z^{(3m+2)/2}}\right), \\
\partial_T^2 F(T_s) &= - 2 (-1)^s \sqrt{Z} \left[ 1 + \sum_{j=3}^{3m-1} F_{2, j}^{(s)} Z^{-j/2} + {\cal O}\left(\frac{1}{Z^{3m/2}}, \frac{\cal T}{Z^{(3m+1)/2}}, \frac{{\cal T}^2}{Z^{(3m+2)/2}}\right) \right], \\
\partial_T^3 F(T_s) &= -2 + \sum_{j=3}^{3m-1} F_{3, j}^{(s)} Z^{-j/2} + {\cal O}\left(\frac{1}{Z^{3m/2}}, \frac{\cal T}{Z^{(3m+1)/2}}, \frac{{\cal T}^2}{Z^{(3m+2)/2}}\right), \\
\begin{split}
\partial_T^k F(T_s) &= (-1)^{k-1} (-1)^{sk} (k-1)! Z^{-k/2} + \sum_{j=k+1}^{3m+k-1} 
F_{k, j}^{(s)} Z^{-j/2} \\
&\phantom{=}+ {\cal O}\left(\frac{1}{Z^{(3m+k)/2}}, \frac{\cal T}{Z^{(3m+k+1)/2}}, \frac{{\cal T}^2}{Z^{(3m+k+2)/2}}\right), \quad j>3.
\end{split}
\end{align}
Furthermore, for large $r$ the coefficients satisfy $ F_{k, j}^{(s)} = \mathcal{O}\left( {\cal T}^{j \mathrm{mod}(3)}\right) $ for $k=0, 2$, and $F_{k, j}^{(s)} = \mathcal{O}\left( {\cal T}^{(j-k) \mathrm{mod}(3)}\right) $ for $k \geq 3$.
\end{lemma}

Based on Lemma \ref{exp-F}, 
we may obtain the higher order expansion using the general formula
\begin{align}
\begin{split}
\frac{1}{2i\pi}\int e^{rF(T)}dT&=
\frac{e^{r F(T_s)}}{2 i \pi} \int \exp\left( r \sum_{j=2}^\infty \frac{f_j}{j!} (T-T_s)^j \right) dT \quad \left(\, \mbox{where}\; f_j:= \partial_T^j F(T_s)\right) \\
&= \frac{e^{r F(T_s)}}{2i\pi} \int e^{ \frac{ r f_2}{2} (T-T_s)^2 } 
\left[ 1 + \sum_{k=3}^\infty S_k\left(0, 0, \frac{r f_3}{3!}, \frac{r f_4}{4!}, \dots, 
\frac{ r f_k}{k!}\right) (T-T_s)^k \right] \\
&= \pm \frac{e^{r F(T_s)}}{2i\pi} \left( 
\frac{\Gamma\left(1/2\right)}{\left(-\frac{rf_2}{2}\right)^{1/2}} + 
\sum_{k=2}^\infty S_{2k}\left(0, 0, \frac{r f_3}{3!}, \frac{r f_4}{4!}, \dots, 
\frac{ r f_{2k} }{(2k)!} \right) \frac{\Gamma\left((2k+1)/2\right)}{\left(-\frac{rf_2}{2}\right)^{(2k+1)/2}} \right) \label{subs}
\end{split}
\end{align}
where the overall sign is related to the contour involved, and the $S_{k}$ are the  polynomials defined by 
\begin{equation}
\exp\left( \sum_{j=1}^\infty x_j z^j \right) = 1 + \sum_{j=1}^\infty S_k(x_1, x_2, \dots, x_k) z^k .
\end{equation}
Note that only the even powered terms contribute after the Gaussian integrals, which all come form the formula
\begin{equation}
\int_{-\infty}^\infty{\rm e}^{-x^2}x^j dx=\Gamma\left(\frac{1+j}{2}\right).
\end{equation}

We require a lemma to control the growth of the $S_k(0, 0, r f_3 / 3!, rf_4/4!, \dots, r f_k/k!) $ for large $r$ and $Z$:
\begin{lemma} \label{poly-lemma}
We find for an integer $M\geq 0$ that 
\begin{align}
\begin{split}
S_{3M+1}\left(0, 0, \frac{r f_3}{3!}, \frac{r f_4}{4!}, \dots \right)  &= 
{\cal O}\left( \frac{ r^M }{Z^{2}} \right), \\
S_{3M+2}\left(0, 0, \frac{r f_3}{3!}, \frac{r f_4}{4!}, \dots\right)  &= 
{\cal O}\left( \frac{ r^M }{Z^{5/2}} \right), \\
S_{3M+3}\left(0, 0, \frac{r f_3}{3!}, \frac{r f_4}{4!}, \dots\right)  &= 
{\cal O}\left( r^{M+1}  \right).
\end{split}
\end{align}
\end{lemma}

\begin{proof}
Using Proposition \ref{exp-F} we see that the leading terms in $S_{k}$ are those with the highest powers of $f_3$ possible, therefore we have 
\begin{align}
\begin{split}
S_{3M+1} &= r^M \frac{f_3^{M-1} f_4}{M (3!)^{M-1} 4!}  + \dots \\
S_{3M+2} &= r^M \frac{f_3^{M-1} f_5}{M (3!)^{M-1} 5!} + \dots \\
S_{3M+3} &= r^{M+1} \frac{ f_3^{M+1} }{ (3!)^{M+1} } +\dots.
\end{split}
\end{align}
The terms left off are lower order in $r$ and higher order in $1/Z$.  
An application of the $f_3$, $f_4$, and $f_5$ entries of Proposition \ref{exp-F} gives the result.
\end{proof}

The series in \eqref{subs} does not converge in most cases, and one must use the original quantity (not the series expansion) to estimate the error.
We can use that ${\rm e}^{rF(T)}$ is analytic at $T=T_{\bf 1}$ and therefore
\begin{equation}\label{error1}
\begin{split}
{\rm e}^{rF(T)-rF(T_1)-\frac{r}{2}F''(T_1)(T-T_1)^2}-\left[1+
\sum_{k=3}^{6M} S_k(0, 0, \frac{rf_3}{3!}, \frac{rf_4}{4!}, \dots, \frac{r f_k}{k!}) (T- T_1)^k \right] \hspace{.8in}\\
=\mathcal{O}\left( \frac{r^{2M}}{Z^2} (T-T_1)^{6M+1}, \frac{r^{2M}}{Z^{5/2}} (T-T_1)^{6M+2} , r^{2M+1} (T-T_1)^{6M+3}, \frac{r^{2M+1}}{Z^2} (T-T_1)^{6M+4} ,
\right. \\ \left.
 \frac{r^{2M+1}}{Z^{5/2}} (T-T_1)^{6M+5} , r^{2M+2} (T-T_1)^{6M+6} \right) 
\end{split}
\end{equation}
on a finite disk around $T_1$ of order $n^{-1/3}$ (the radius of convergence is determined by that of $F(T)$).  The right-hand side of this expression follows from Lemma \ref{poly-lemma}.  Note also that only the even terms will contribute to the Gaussian integration.
  Then the error is given by
\begin{equation} \label{71eq}
\begin{split}
&\left( - \frac{r}{2} F''(T_1) \right)^{1/2} 
\int {\rm e}^{\frac{r}{2}F''(T_1)(T-T_1)^2}
\left[ 
\mathcal{O}\left( \frac{r^{2M}}{Z^2} (T-T_1)^{6M+1}, \frac{r^{2M}}{Z^{5/2}} (T-T_1)^{6M+2},  r^{2M+1} (T-T_1)^{6M+3}, \right. \right. \\ & \phantom{\left( - \frac{r}{2} F''(T_1) \right)^{1/2} 
\int {\rm e}^{\frac{r}{2}F''(T_1)(T-T_1)^2}} \left. \left.
\frac{r^{2M+1}}{Z^2} (T-T_1)^{6M+4},  \frac{r^{2M+1}}{Z^{5/2}} (T-T_1)^{6M+5}, r^{2M+2} (T-T_1)^{6M+6} \right) 
\right]
 d T \\ 
& = \mathcal{O}\left( \frac{r^{2M}}{Z^{5/2}} \left( -\frac{r}{2} f_2 \right)^{-3M-1}, \frac{r^{2M+1}}{Z^2} \left( -\frac{r}{2} f_2 \right)^{-3M-2}, r^{2M+2}\left( -\frac{r}{2} f_2 \right)^{-3M-3} \right)
=
\mathcal{O}\left( \left( \frac{r}{\zeta^3}\right)^{M+1} \right) ,
\end{split}
\end{equation}
where in the last step we used Proposition \ref{exp-F}.
Note that 
\eq
\frac{r}{n^{3(\frac 23-\mu)}} = \mathcal{O}\left(n^{\gamma-2 +3\mu}\right) < \mathcal{O}\left(n^{\gamma-2+(2-2\gamma)}\right) = \mathcal{O}(n^{-\gamma}).
\endeq

We will give here a general argument for the existence of the expansion to arbitrary order of $v_1^{(r)}$.
A similar argument will give the existence of the expansion to arbitrary order of $v_2^{(r)}$.  
\begin{equation} \label{eq185}
\begin{split}
v_1^{(r)}(\zeta) = \mbox{Ai}_{\mathcal{C}_1}^{(r)}(\zeta) = & 
- \frac{r^{(r+1)/3}}{2 \sqrt{\pi} i} \frac{ e^{r F(T_1)}}{\left( -\frac{r}{2} F''(T_1) \right)^{1/2} }\times \\
& \times\left( 1 + \sum_{k=2}^{3M} S_{2k} \frac{ \Gamma\left( (2k+1)/2\right) }{\sqrt{\pi}} \left( -\frac{r}{2} F''(T_1) \right)^{-k} + \mathcal{O}\left( \left( \frac{r}{\zeta^3}\right)^{M+1} \right) \right).
\end{split}
\end{equation}
Each of the terms in (\ref{eq185}) have a truncated expansion whose error is dominated by $\mathcal{O}\left( r^{4M+4} / \zeta^{3M + 3} \right)$.  

To see this we begin with 
\begin{align}
\begin{split}
\left( -\frac{r}{2} F''(T_1) \right)^{-1/2} &= \frac{1}{i r^{1/3} \zeta^{1/4}} \left[ 1 - \sum_{j=3}^{6M+5} F_{2, j}^{(1)} \frac{ r^{j/3} }{\zeta^{j/2}} + 
\mathcal{O}\left( \frac{r^{2M+2}}{\zeta^{3M+3}}, \frac{r^{2M+3}}{\zeta^{3M+9/2}}\right) \right]^{-1/2}
\\
&= \frac{1}{i r^{1/3} \zeta^{1/4}} \left[ 1 + \sum_{j=3}^{6M+5} R_j(0, 0, F_{2, 3}^{(1)} r, F_{2, 4}^{(1)} r^{4/3}, \dots) \zeta^{-j/2} + \mathcal{O}\left( \frac{r^{2M+2}}{\zeta^{3M+3}}, \frac{r^{2M+3}}{\zeta^{3M+9/2}}\right) \right],
\end{split}
\end{align}
where
\begin{equation}
\left( 1 + \sum_{j=1}^\infty x_j z^j \right)^{-1/2} = 1 + \sum_{j=1}^\infty R_j(x_1, x_2, \dots, x_j) z^j ,
\end{equation}
and we have used Proposition \ref{exp-F}.
One checks that indeed $ R_j = \mathcal{O}\left( \frac{ r^{\lfloor j/3 \rfloor}}{\zeta^{j/2}} \right)$.  

Next we examine
\begin{align}
\begin{split}
e^{r F(T_1)} &= \exp\left( -\frac{2}{3} \zeta^3 + r \pi i - \frac{r}{3} \log(r) + r \log( \zeta^{1/2} - \tau) + \sum_{j=3}^{6M} r^{1+j/3} F_{0, j}^{(1)} \zeta^{-j/2} + \mathcal{O}\left( \frac{r^{2M+3}}{\zeta^{3M+3}}, \frac{r^{2M+4}}{\zeta^{3M+9/2}} \right) \right)\\
&= (-1)^r (\zeta^{1/2} - \tau)^r r^{-r/3} \exp\left( -\frac{2}{3}\zeta^{3/2} + \sum_{j=3}^{6M+6} r^{1+j/3} F_{0, j}^{(1)} \zeta^{-j/2} \right) \left[ 1 + \mathcal{O}\left( \frac{r^{2M+3}}{\zeta^{3M+3}}, \frac{r^{2M+4}}{\zeta^{3M+9/2}} \right) \right].
\label{eq-186}
\end{split}
\end{align}
We control each of the remaining terms in the exponent of (\ref{eq-186}) 
separately:
\begin{equation}\label{source-j}
e^{ r^{1 + j/3} F_{0, j}^{(1)} \zeta^{-j/2} } = \sum_{k=0}^{\lceil (6M+6)/j \rceil -1} 
\frac{1}{k!} \left( r^{1+j/3} F_{0, j}^{(1)} \zeta^{-j/2} \right)^k + 
\mathcal{O}\left( \left( \frac{r^{\lfloor j/3 \rfloor + 1} }{ \zeta^{j/2}} \right)^{\lceil (6M+6)/j \rceil} \right),
\end{equation}
where we have used Proposition \ref{exp-F} and the identity $j/3 - ( j \mathrm{mod}(3) )/3 = \lfloor j/3 \rfloor$.  
This error will be the dominant one (for some choices of $j$).  We have 
\begin{equation}
\frac{r^{\lfloor j/3 \rfloor +1}}{\zeta^{j/2}} 
\end{equation}
is largest when $j \mbox{mod}(3) = 0$.  Suppose $j=3 l$, then we have the error arising
from the expansion (\ref{source-j}) is 
\begin{equation}
\mathcal{O}\left( \frac{r^{(l+1)\lceil(2M+2)/l\rceil}}{\zeta^{ \frac{3}{2} l \lceil(2M+2)/l \rceil }} \right) .
\end{equation}
Noting that 
\begin{equation}
\gamma (l+1) - 3 l / 2 <  2 \gamma - 3/2 , \quad \mbox{for} \quad l > 1,
\end{equation}
we conclude that the dominant error contribution to (\ref{eq-186}) from the expansions (\ref{source-j}) for $3 \leq j \leq 6M+5 $ is 
\begin{equation} \label{dom-err}
\mathcal{O}\left( \frac{ r^{4M+4}}{\zeta^{3M+3}} \right) .
\end{equation}  
The remaining terms in (\ref{eq185}) may be analyzed in a similar way, the error term (\ref{dom-err}) remains the dominant one, and our conclusion is that $v_1^{(r)}$ may be written as a finite collection of terms of decreasing order plus an error term of the form
\begin{equation}
\mathcal{O}\left( \frac{r^{4M+4}}{\zeta^{3M+3}}\right).
\end{equation}
This is small (in large $n$) for $\gamma < \frac{1}{8}$.  A similar analysis may be carried out for $v_2^{(r)}(\zeta)$ to derive an identical error bound.
The entries of the second row are then computed directly by taking a derivative in $\zeta$ of these expressions; the entries of the third row are also computed directly by substituting for $r = r-1$ in the formulas for the first row.
 This completes the proof that the expansions in (\ref{arbitrary-v1})--(\ref{arbitrary-32}) in Proposition \ref{Ar-col12} are asymptotic.

\subsection{Proof of Proposition \ref{Ar-col3}}
Here we prove the existence of general asymptotic expansions for the 
third column of $\mathcal{A}_{(r)}(\zeta)$ (see \eqref{Ar13}).  
To obtain the expansion of $\mathcal{A}_{(r,13)}$ we 
use the definition \eqref{148} in terms of Cauchy transforms.
\begin{align}
\begin{split}
v_3^{(r)}(\zeta)&=\frac{{\rm e}^{-\tau\zeta}}{2i\pi}\left(
\int_{{\cal C}_6}\frac{{\rm Ai}_{{\cal C}_6}^{(r)}(t){\rm e}^{\tau t}dt}{t-\zeta}-
\int_{{\cal C}_2}\frac{{\rm Ai}_{{\cal C}_2}^{(r)}(t){\rm e}^{\tau t}dt}{t-\zeta}\right)
\\&=\frac{{\rm e}^{-\tau\zeta}}{2i\pi}\left(
\int_{{\cal C}_5}\frac{{\rm Ai}_{{\cal C}_6}^{(r)}(t){\rm e}^{\tau t}dt}{t-\zeta}-
\int_{{\cal C}_3}\frac{{\rm Ai}_{{\cal C}_2}^{(r)}(t){\rm e}^{\tau t}dt}{t-\zeta}+
\int_{{\cal C}_4}\frac{{\rm Ai}_{{\cal C}_1}^{(r)}(t){\rm e}^{\tau t}dt}{t-\zeta}\right).
\end{split}
\end{align}
The above equality is simply the recombination of contours.  Then, from Definition \ref{155}, we can see that the numerator is a total derivative, i.e. ${\rm Ai}_{{\cal C}_6}^{(r)}(\zeta){\rm e}^{\tau t}=\partial_\zeta\left({\rm Ai}_{{\cal C}_6}^{(r-1)}(\zeta){\rm e}^{\tau t}\right)$.  Performing integration by parts, we get
\begin{equation}
\begin{split}
(\mbox{above})=&
\frac{{\rm e}^{-\tau\zeta}}{2i\pi}\left(
\left.\frac{{\rm Ai}_{{\cal C}_6}^{(r-1)}(t){\rm e}^{\tau t}}{t-\zeta}\right|_{{\cal C}_5}-
\left.\frac{{\rm Ai}_{{\cal C}_2}^{(r-1)}(t){\rm e}^{\tau t}}{t-\zeta}\right|_{{\cal C}_3}+
\left.\frac{{\rm Ai}_{{\cal C}_1}^{(r-1)}(t){\rm e}^{\tau t}}{t-\zeta}\right|_{{\cal C}_4}\right)
\\ &+\frac{{\rm e}^{-\tau\zeta}}{2i\pi}\left(
\int_{{\cal C}_5}\frac{{\rm Ai}_{{\cal C}_6}^{(r-1)}(t){\rm e}^{\tau t}dt}{(t-\zeta)^2}-
\int_{{\cal C}_3}\frac{{\rm Ai}_{{\cal C}_2}^{(r-1)}(t){\rm e}^{\tau t}dt}{(t-\zeta)^2}+
\int_{{\cal C}_4}\frac{{\rm Ai}_{{\cal C}_1}^{(r-1)}(t){\rm e}^{\tau t}dt}{(t-\zeta)^2}\right).
\end{split}
\end{equation}
The first three terms are evaluated at the endpoints of the contours, ${\cal C}_3, {\cal C}_4$, and ${\cal C}_5$.   Each contour starts from the same point and goes to infinity.   By ${\rm Ai}^{(r)}_{{\cal C}_1}-{\rm Ai}^{(r)}_{{\cal C}_2}+{\rm Ai}^{(r)}_{{\cal C}_6}\equiv 0$ and the decay property of each integrand at the corresponding infinity, the first three terms vanish. 

We can perform the integration by parts on the remaining three integrals recursively to finally get
\begin{equation}
v_3^{(r)}(\zeta)=r!\frac{{\rm e}^{-\tau\zeta}}{2i\pi}\left(
\int_{{\cal C}_5}\frac{{\rm Ai}_{{\cal C}_6}^{(0)}(t){\rm e}^{\tau t}dt}{(t-\zeta)^{r+1}}-
\int_{{\cal C}_3}\frac{{\rm Ai}_{{\cal C}_2}^{(0)}(t){\rm e}^{\tau t}dt}{(t-\zeta)^{r+1}}+
\int_{{\cal C}_4}\frac{{\rm Ai}_{{\cal C}_1}^{(0)}(t){\rm e}^{\tau t}dt}{(t-\zeta)^{r+1}}\right).
\end{equation}
We intend to obtain the asymptotic expansion of the above in large $\zeta$.   
We define $\Delta_{M}$ such that
\begin{equation}
v_{3}^{(r)}(\zeta)=\frac{r!}{(-1)^{r+1}}\frac{{\rm e}^{-\tau\zeta}}{2 i\pi} 
\left[\Delta_M+\sum_{(A,B)} \int_{{\cal C}_{B}}{\rm Ai}_{{\cal C}_{A}}^{(0)}(t)\frac{{\rm e}^{\tau t}}{\zeta^{r+1}}\left(1+\frac{r+1}{\zeta}t+\cdots+\frac{(r+M)!}{r!M!}\frac{t^M}{\zeta^M}\right)dt\right],
\end{equation}
that is,
\begin{equation}
\Delta_M:=\sum_{(A,B)} \int_{{\cal C}_{B}}{\rm Ai}_{{\cal C}_{A}}^{(0)}(t){\rm e}^{\tau t}\left[\frac{(-1)^{r+1}}{(t-\zeta)^{r+1}}-\left(\frac{1}{\zeta^{r+1}}+\frac{r+1}{\zeta^{r+2}}t+\cdots+\frac{(r+M)!}{r!M!}\frac{t^M}{\zeta^{r+M+1}}\right)\right]d t,
\end{equation}
where $(A,B)=(6,5),(-2,3),(1,4)$.

We divide the above $t$-integral (on ${\cal C}_{B}$) into two parts: one for $|t|<L$ and the other for $|t|\geq L$ where $L$ is set by 
\begin{equation}
L= n^{\delta},\qquad 0<\delta<1/12.
\end{equation}
There exists $n_{0}>0$ such that, for $n>n_{0}$,
\begin{equation}
\left|{\rm Ai}^{(0)}_{{\cal C}_A}(t)\right|=\frac{1}{2\pi}\frac{{\rm e}^{-\frac23 |t|^{3/2}}}{|t|^{1/4}}\left(1+{\cal O}\left(\frac{1}{|t|^{3/2}}\right)\right),\quad \mbox{$t\in{\cal C}_B$, $|t|\geq L=n^{\delta}$}.
\end{equation}

For $r=\mathcal{O}(n^\gamma)$, $0\leq\gamma<1/12$, and $|\zeta|=\mathcal{O}(n^{1/6})$, there exists $n_{0}$ such that the following crude estimate holds for all $n>n_{0}$:
\begin{equation}
\left|\frac{(-1)^{r+1}}{(t-\zeta)^{r+1}}-\left(\frac{1}{\zeta^{r+1}}+\frac{r+1}{\zeta^{r+2}}t+\cdots+\frac{(r+M)!}{r!M!}\frac{t^M}{\zeta^{r+M+1}}\right)\right|<{\rm e}^{ |t|},\qquad |t|\geq L=n^{\delta}.
\end{equation}
This is true since each term in the left hand side is bounded by ${\cal O}( |t|^{M})$ for $n$ large enough.

Using the above two estimates, we conclude that, for $n>n_{0}$ (where $n_{0}$ satisfies the above two characterizations), the contribution of $|t|\geq L$ to $\Delta_{M}$ is bounded by
\begin{align}
\begin{split}
\left|\Delta_{M}\right|_{\text{from $|t|\geq L$}}&<3|\zeta|^{r+1}\int_{|t|>L} \frac{1}{2\pi}\frac{{\rm e}^{-\frac23|t|^{3/2}}}{|t|^{1/4}}{\rm e}^{(|\tau|+1)|t|} |d t|
\\
&<3|\zeta|^{r+1}\int_{L}^{\infty} {\rm e}^{-\frac13 t ^{3/2}} dt<6|\zeta|^{r+1}{\rm e}^{-\frac13 L^{3/2}}\sim 3\,n^{(r+1)/6}{\rm e}^{-\frac13 n^{3\delta/2}}. \label{t>L}
\end{split}
\end{align}
The second inequality holds if we choose $n_{0}$ large enough.  The third inequality is obtained by $\int_{L}^{\infty}{\rm e}^{-\frac13 t^{3/2}}dt<\int_{L}^{\infty}{\rm e}^{-\frac13 t^{3/2}}\frac32 t^{1/2}d t$ for $L$ large enough.  The last estimate is for $|\zeta|\sim n^{1/6}$.

Now we consider $|t|<L$. 

An error bound of a finite Taylor expansion for a general analytic function $f$ is given by
\begin{equation}
\begin{split}
\left|f(x+y)-\left(f(x)+y f'(x)+...\frac{y^M}{M!} f^{(M)}(x)\right)\right|&=\left|\int_0^yd s_1\int_0^{s_1}d s_2\cdots \int_0^{s_M}d s_{M+1} f^{(M+1)}(x+s_{M+1})\right|
\\&\leq \frac{y^{M+1}}{(M+1)!} \max_{s\in [0,y]}\left| f^{(M+1)}(x+s ) \right|.
\end{split}
\end{equation}
Let us apply this to our case, $f(x)=1/x^{r+1}$ and $y=-t$.  We get, for $|t|<L$,
\begin{align}
\begin{split}
\left|\frac{(-1)^{r+1}}{(t-\zeta)^{r+1}}-\left(\frac{1}{\zeta^{r+1}}+\frac{r+1}{\zeta^{r+2}}t+\cdots+\frac{(r+M)!}{r!M!}\frac{t^M}{\zeta^{r+M+1}}\right)\right|\leq 
\frac{|t|^{M+1}}{(M+1)!}\max_{s\in [-t,0]}\left|\frac{(r+M+1)!}{r!(\zeta+s)^{r+M+2}}\right|
\\\leq  \frac{|t|^{M+1}}{(M+1)!}\frac{(r+M+1)!}{r!(|\zeta|-|L|)^{r+M+2}}=\left({r+M+1}\atop M+1\right)\frac{|t|^{M+1}}{|\zeta|^{r+M+2}}\frac{1}{(1-|L/\zeta|)^{r+M+2}}.
\end{split}
\end{align}
The last factor is bounded by 2 for $n$ large enough because
\begin{equation}
\frac{1}{(1-|L/\zeta|)^{r+M+1}}=\left(1-\left|\frac{L}{\zeta}\right|\right)^{\left|\frac{\zeta}{L}\right|\times \left|\frac{L}{\zeta}(r+M+1)\right|} \to 1
\end{equation}
 as $n\to\infty$.  Here we have used that $L r/\zeta \sim n^{\delta + \gamma-\frac16}\to 0$ for $\delta<1/12$.

The contribution of $|t|<L$ to $\Delta_{M}$ is then bounded by
\begin{align}
\begin{split}
\left|\Delta_{M}\right|_{\text{from $t< L$}}&<\left({r+M+1}\atop {M+1}\right)\frac{2}{|\zeta|^{r+M+2}}\sum_{(A,B)}\int_{t\in {\cal C}_B, t<L} \left|{\rm Ai}_{{\cal C}_A}^{(0)}(t)\right|{\rm e}^{|\tau t|}|t|^{M+1}|d t|
\\&<\left({r+M+1}\atop {M+1}\right)\frac{6}{|\zeta|^{r+M+2}}\int_{0}^{\infty} {\rm Ai}(t){\rm e}^{|\tau| t} t^{M+1}d t < \text{const.} \frac{r^{M+1}}{|\zeta|^{r+M+2}}, \label{t<L}
\end{split}
\end{align} 
where the constant factor depends only on $M$ and $\tau$, not on $n$. 

From \eqref{t>L} and \eqref{t<L} we get 
\begin{equation}
\left|\Delta_{M}\right|= {\cal O}\left( \frac{r^{M+1}}{|\zeta|^{r+M+2}}\right)
\end{equation}
and we obtain the asymptotic expansion
\begin{equation}
v_3^{(r)}(\zeta)=\frac{r!}{(-1)^{r+1}}\frac{{\rm e}^{-\tau \zeta}}{2 i\pi}\left(\sum_{j=0}^{M} \frac{\eta_{j}(\tau)}{\zeta^{r+j+1}}+  {\cal O}\left( \frac{r^{M+1}}{|\zeta|^{r+M+2}}\right) \right),\qquad \zeta\in \partial {\mathbb D}_{c},
\end{equation}
where 
\begin{equation}
\eta_{j}(\tau):=\sum_{(A,B)}\int_{{\cal C}_B}{\rm Ai}_{{\cal C}_A}^{(0)}(t){\rm e}^{\tau t} t^{j} d t.
\end{equation}
This concludes the proof.

\subsection{The leading-order terms in the expansion of 
\texorpdfstring{$\bs{\mathcal{A}_{(r)}(\zeta)}$}{13212}}
\label{Ar-explicit}
We include here the explicit form of the first few terms of the asymptotic 
expansions of the entries of $\mathcal{A}_{(r)}(\zeta)$ to give an idea of 
their form.  Note that more terms would be needed to carry out the 
Schlesinger calculations explicitly, but we do not need the exact formulas 
for our results.  We find:
\begin{align}
\begin{split}
\mathcal{A}_{(r,11)}(\zeta) & = \frac{(-1)^r}{2 \sqrt{\pi}} \left(\zeta^{1/2} - \tau\right)^r \zeta^{ - 1/4} e^{-\frac{2}{3} \zeta^{3/2}} \times \\
  & \hspace{.2in}\left( 1 + \left(-\frac{r^2}{4} +\frac{r}{2}-\frac{5}{48} \right) \frac{1}{\zeta^{3/2}} + \left( - \frac{r^2\tau}{2} + \frac{3 r\tau}{4} \right) \frac{1}{\zeta^2} + \frac{1}{32} \frac{r^4}{\zeta^3} + \mathcal{O}\left( \frac{r^2}{\zeta^{5/2}}\right) + \mathcal{O}\left( \frac{r^3}{\zeta^3}\right) \right), \end{split}\\
\begin{split}
\mathcal{A}_{(r,21)}(\zeta) &  = \frac{(-1)^r}{2\sqrt{\pi}} (\zeta^{1/2}-\tau)^r \zeta^{ 1/4} e^{-\frac{2}{3} \zeta^{3/2}} \times \\
  & \phantom{=\frac{(-1)^r}{2\sqrt{\pi}} } \left( -1 + \left( \frac{r^2}{4} - \frac{7}{48} \right) \frac{1}{\zeta^{3/2}} + \left( \frac{r^2\tau}{2} - \frac{r\tau}{4} \right) \frac{1}{\zeta^2} - \frac{1}{32} \frac{r^4}{\zeta^3} + \mathcal{O}\left(\frac{r^2}{\zeta^{5/2}}\right) + \mathcal{O}\left(\frac{r^3}{\zeta^3}\right) \right),
\end{split}
\\
\mathcal{A}_{(r,31)}(\zeta) & = \frac{(-1)^{r-1}}{2\sqrt{\pi}} (\zeta^{1/2}-\tau)^r e^{-\frac{2}{3}\zeta^{3/2}} \left( \frac{1}{\zeta^{3/4}} + \frac{\tau}{\zeta^{5/4}} + \frac{\tau^2}{\zeta^{7/4}} + \mathcal{O}\left(\frac{r^2}{\zeta^{9/4}}\right) \right)\,, \\
\begin{split}
\mathcal{A}_{(r,12)}(\zeta) &  = -\frac{1}{2\sqrt{\pi} i} \left(\zeta^{1/2} + \tau\right)^r \zeta^{-1/4} e^{\frac{2}{3}\zeta^{3/2} } \times \\ 
 & \left( 1 + \left( \frac{r^2}{4} - \frac{r}{2} + \frac{5}{48} \right) \frac{1}{\zeta^{3/2}} +\left( - \frac{r^2 \tau}{2}+ \frac{3 r\tau}{4} \right) \frac{1}{\zeta^2} + \frac{1}{32} \frac{r^4}{\zeta^3} + \mathcal{O}\left( \frac{r^2}{\zeta^{5/2}}\right) + \mathcal{O}\left( \frac{r^3}{\zeta^3}\right)\right)\,, 
\end{split} \\
\begin{split}
\mathcal{A}_{(r,22)}(\zeta) & = -\frac{1}{2\sqrt{\pi} i} \left(\zeta^{1/2}+\tau\right)^r \zeta^{1/4} e^{\frac{2}{3} \zeta^{3/2} } \times 
\\
&\phantom{=}
\left( 1 + \left(\frac{r^2}{4} - \frac{7}{48}\right) \frac{1}{\zeta^{3/2}} + \left( - \frac{r^2 \tau}{2} + \frac{r\tau}{4} \right) \frac{1}{\zeta^2} + \frac{1}{32} \frac{r^4}{\zeta^3} + \mathcal{O}\left( \frac{r^2}{\zeta^{5/2}}\right)+\mathcal{O}\left( \frac{r^3}{\zeta^3}\right) \right), 
\end{split}\\
\mathcal{A}_{(r,32)}(\zeta) &  = - \frac{1}{2\sqrt{\pi} i} \left(\zeta^{1/2}+\tau\right)^{r} e^{\frac{2}{3}\zeta^{3/2} } \left( \frac{1}{\zeta^{3/4}} 
-\frac{\tau}{\zeta^{5/4}} + \frac{\tau^2}{\zeta^{7/4}} 
+ \mathcal{O}\left(\frac{r^2}{\zeta^{9/4}}\right) \right),\\
\mathcal{A}_{(r,13)}(\zeta) & = \mbox{Ai}_{\mathcal{C}_3}^{(r)}(\zeta) = \frac{(-1)^{r+1}}{2\pi i} \zeta^{-r} r! e^{-\tau \zeta} \eta_0(\tau) \left( \frac{1}{\zeta} + \mathcal{O}\left(\frac{r}{\zeta^2}\right) \right), \\
\mathcal{A}_{(r,23)}(\zeta) & = \partial_\zeta \mbox{Ai}_{\mathcal{C}_3}^{(r)}(\zeta) = 
\frac{(-1)^{r+1}}{2\pi i} \zeta^{-r} r! e^{-\tau\zeta} \eta_0(\tau) \left( -\frac{\tau}{\zeta} + \mathcal{O}\left( \frac{r}{\zeta^2}\right) \right),\\
\mathcal{A}_{(r,33)}(\zeta) & = \mbox{Ai}_{\mathcal{C}_3}^{(r-1)}(\zeta) = 
\frac{(-1)^r }{2\pi i} \zeta^{-r} (r-1)! e^{-\tau\zeta} \eta_0(\tau) \left( 1 + \mathcal{O}\left( \frac{r}{\zeta}\right) \right),
\end{align} 
where $\eta_0(\tau)$ is given by \eqref{eta-of-tau}.

\end{appendix}

%

\bibliographystyle{alpha}

\end{document}

%% file: crit_multicut1.pdf_t
\begin{picture}(0,0)%
\includegraphics{crit_multicut1.pdf}%
\end{picture}%
\setlength{\unitlength}{3947sp}%
\begingroup\makeatletter\ifx\SetFigFont\undefined%
\gdef\SetFigFont#1#2#3#4#5{%
  \reset@font\fontsize{#1}{#2pt}%
  \fontfamily{#3}\fontseries{#4}\fontshape{#5}%
  \selectfont}%
\fi\endgroup%
\begin{picture}(15699,15905)(-161,-10594)
\put(4801,-1786){\makebox(0,0)[lb]{\smash{{\SetFigFont{17}{20.4}{\rmdefault}{\mddefault}{\updefault}{\color[rgb]{0,0,0}$\beta_1$}%
}}}}
\put(8251,-1786){\makebox(0,0)[lb]{\smash{{\SetFigFont{17}{20.4}{\rmdefault}{\mddefault}{\updefault}{\color[rgb]{0,0,0}$\beta_2$}%
}}}}
\put(13426,-1786){\makebox(0,0)[lb]{\smash{{\SetFigFont{17}{20.4}{\rmdefault}{\mddefault}{\updefault}{\color[rgb]{0,0,0}$\beta$}%
}}}}
\put(2176,-1786){\makebox(0,0)[lb]{\smash{{\SetFigFont{17}{20.4}{\rmdefault}{\mddefault}{\updefault}{\color[rgb]{0,0,0}$\alpha_1$}%
}}}}
\put(5626,-1786){\makebox(0,0)[lb]{\smash{{\SetFigFont{17}{20.4}{\rmdefault}{\mddefault}{\updefault}{\color[rgb]{0,0,0}$\alpha_2$}%
}}}}
\put(10801,-1786){\makebox(0,0)[lb]{\smash{{\SetFigFont{17}{20.4}{\rmdefault}{\mddefault}{\updefault}{\color[rgb]{0,0,0}$\alpha$}%
}}}}
\put(7951,4964){\makebox(0,0)[lb]{\smash{{\SetFigFont{20}{24.0}{\rmdefault}{\mddefault}{\updefault}{\color[rgb]{0,0,0}$L$}%
}}}}
\put(6076,-4186){\makebox(0,0)[lb]{\smash{{\SetFigFont{20}{24.0}{\rmdefault}{\mddefault}{\updefault}{\color[rgb]{0,0,0}$L$}%
}}}}
\put(7876,-8161){\makebox(0,0)[lb]{\smash{{\SetFigFont{20}{24.0}{\rmdefault}{\mddefault}{\updefault}{\color[rgb]{0,0,0}${\mathbb D}_c$}%
}}}}
\put(8251,-5611){\makebox(0,0)[lb]{\smash{{\SetFigFont{20}{24.0}{\rmdefault}{\mddefault}{\updefault}{\color[rgb]{0,0,0}${\mathbb D}(\beta,\delta)$}%
}}}}
\put(12001,-1336){\makebox(0,0)[lb]{\smash{{\SetFigFont{20}{24.0}{\rmdefault}{\mddefault}{\updefault}{\color[rgb]{0,0,0}$\Omega^+_{\text{main}}$}%
}}}}
\put(12001,-1936){\makebox(0,0)[lb]{\smash{{\SetFigFont{20}{24.0}{\rmdefault}{\mddefault}{\updefault}{\color[rgb]{0,0,0}$\Omega^-_{\text{main}}$}%
}}}}
\put(12976,764){\makebox(0,0)[lb]{\smash{{\SetFigFont{25}{30.0}{\rmdefault}{\mddefault}{\updefault}{\color[rgb]{0,0,0}$L$}%
}}}}
\put(3901,-6361){\makebox(0,0)[lb]{\smash{{\SetFigFont{20}{24.0}{\rmdefault}{\mddefault}{\updefault}{\color[rgb]{0,0,0}$\Omega^+_{\text{main}}$}%
}}}}
\put(3901,-8236){\makebox(0,0)[lb]{\smash{{\SetFigFont{20}{24.0}{\rmdefault}{\mddefault}{\updefault}{\color[rgb]{0,0,0}$\Omega^-_{\text{main}}$}%
}}}}
\put(5476,-4786){\makebox(0,0)[lb]{\smash{{\SetFigFont{20}{24.0}{\rmdefault}{\mddefault}{\updefault}{\color[rgb]{0,0,0}$\Omega^+_L$}%
}}}}
\put(5551,-9661){\makebox(0,0)[lb]{\smash{{\SetFigFont{20}{24.0}{\rmdefault}{\mddefault}{\updefault}{\color[rgb]{0,0,0}$\Omega^-_L$}%
}}}}
\put(6526,4289){\makebox(0,0)[lb]{\smash{{\SetFigFont{20}{24.0}{\rmdefault}{\mddefault}{\updefault}{\color[rgb]{0,0,0}$\Omega^+_L$}%
}}}}
\put(6526,2714){\makebox(0,0)[lb]{\smash{{\SetFigFont{20}{24.0}{\rmdefault}{\mddefault}{\updefault}{\color[rgb]{0,0,0}$\Omega^-_L$}%
}}}}
\put(8851,4289){\makebox(0,0)[lb]{\smash{{\SetFigFont{20}{24.0}{\rmdefault}{\mddefault}{\updefault}{\color[rgb]{0,0,0}$\Omega^+_{\text{out}}$}%
}}}}
\put(8851,2714){\makebox(0,0)[lb]{\smash{{\SetFigFont{20}{24.0}{\rmdefault}{\mddefault}{\updefault}{\color[rgb]{0,0,0}$\Omega^+_{\text{out}}$}%
}}}}
\put(3451,-1336){\makebox(0,0)[lb]{\smash{{\SetFigFont{20}{24.0}{\rmdefault}{\mddefault}{\updefault}{\color[rgb]{0,0,0}$\Omega_1^+$}%
}}}}
\put(3451,-1936){\makebox(0,0)[lb]{\smash{{\SetFigFont{20}{24.0}{\rmdefault}{\mddefault}{\updefault}{\color[rgb]{0,0,0}$\Omega_1^-$}%
}}}}
\put(6901,-1336){\makebox(0,0)[lb]{\smash{{\SetFigFont{20}{24.0}{\rmdefault}{\mddefault}{\updefault}{\color[rgb]{0,0,0}$\Omega_2^+$}%
}}}}
\put(6901,-1936){\makebox(0,0)[lb]{\smash{{\SetFigFont{20}{24.0}{\rmdefault}{\mddefault}{\updefault}{\color[rgb]{0,0,0}$\Omega_2^-$}%
}}}}
\end{picture}%